\documentclass[lettersize,journal]{IEEEtran}
\usepackage{cite}
\usepackage{amsmath,amssymb,amsfonts}
\usepackage{algorithmic}
\usepackage{graphicx}
\usepackage{textcomp}
\usepackage{subfig}
\usepackage{dsfont}
\usepackage{color, soul}
\usepackage{amsthm}
\usepackage{color,soul}
\usepackage[normalem]{ulem}

\usepackage{array}
\newcolumntype{M}[1]{>{\centering\arraybackslash}m{#1}}
\usepackage{hhline}

\def\BibTeX{{\rm B\kern-.05em{\sc i\kern-.025em b}\kern-.08em
    T\kern-.1667em\lower.7ex\hbox{E}\kern-.125emX}}

\newcommand{\Pt}{ \widetilde{P}}
\newcommand{\pt}{ \widetilde{p}}
\newcommand{\Yt}{ \widetilde{Y}}
\newcommand{\yt}{ \widetilde{y}}
\newcommand{\Mt}{M_{\Lambda}}
\newcommand{\Mo}{M_{\Omega}}
\newcommand{\Rt}{ R_\Lambda}
\newcommand{\mt}{ \widetilde{m}}
\newcommand{\Pds}{ \mathds{P}}
\newcommand{\Eds}{ \mathds{E}}
\newcommand{\eye}{ \mathds{1}}
\newcommand{\tht}{\hat{\theta}}

\newtheorem{claim}{Claim}

\begin{document}
\title{A theoretical framework for self-supervised MR image reconstruction using sub-sampling via variable density Noisier2Noise}
\author{Charles Millard and Mark Chiew 
\thanks{This work was supported in part by the Engineering and Physical Sciences Research Council, grant EP/T013133/1, by the Royal Academy of Engineering, grant RF201617/16/23, and by the Wellcome Trust, grant 203139/Z/16/Z. The computational aspects of this research were supported by the Wellcome Trust Core Award Grant Number 203141/Z/16/Z and the NIHR Oxford BRC. The views expressed are those of the authors and not necessarily those of the NHS, the NIHR or the Department of Health. This research was undertaken, in part, thanks to funding from the Canada Research Chairs Program.}
\thanks{Charles Millard and Mark Chiew are with the Wellcome Centre for Integrative Neuroimaging, FMRIB, Nuffield Department of Clinical Neurosciences, University of Oxford, Level 0, John Radcliffe Hospital, Oxford, OX3 9DU, UK. Mark Chiew is also with the Department of Medical Biophysics, University of Toronto, Toronto, Canada and Physical Sciences, Sunnybrook Research Institute, Toronto, Canada (email: charles.millard@ndcn.ox.ac.uk and mark.chiew@utoronto.ca).}}

\markboth{IEEE TRANSACTIONS ON COMPUTATIONAL IMAGING, VOL. X, NO. X, XXX XX}{Self-supervised deep learning MRI reconstruction with Noisier2Noise}

\maketitle

\begin{abstract}
In recent years, there has been attention on leveraging the statistical modeling capabilities of neural networks for reconstructing sub-sampled Magnetic Resonance Imaging (MRI) data. Most proposed methods assume the existence of a representative fully-sampled dataset and use fully-supervised training. However, for many applications, fully sampled training data is not available, and may be highly impractical to acquire. The development and understanding of self-supervised methods, which use only sub-sampled data for training, are therefore highly desirable. This work extends the Noisier2Noise framework, which was originally constructed for self-supervised denoising tasks, to variable density sub-sampled MRI data.  We use the Noisier2Noise framework to analytically explain the performance of Self-Supervised  Learning  via  Data  Undersampling  (SSDU), a recently proposed method that performs well in practice but until now lacked theoretical justification. Further, we propose two modifications of SSDU that arise as a consequence of the theoretical developments. Firstly, we propose partitioning the sampling set so that the subsets have the same type of distribution as the original sampling mask. Secondly, we propose a loss weighting that compensates for the sampling and partitioning densities. On the fastMRI dataset we show that these changes significantly improve SSDU's image restoration quality and robustness to the partitioning parameters.

\end{abstract}


\section{Introduction}
\label{sec:introduction}

\IEEEPARstart{T}{he} data acquisition process in Magnetic Resonance Imaging (MRI) consists of traversing a sequence of smooth paths through the Fourier representation of the image, referred to as ``k-space", which is inherently time-consuming. Images can be reconstructed from accelerated, sub-sampled acquisitions by leveraging the non-uniformity of receiver coil sensitivities, referred to as ``parallel imaging" \cite{ra1993fast, Pruessmann1999, Griswold2002, Uecker2014}. Compressed sensing \cite{Donoho2006, Candes2006}, which uses sparse models to reconstruct incoherently sampled data, has also  been widely applied to MRI \cite{Lustig2007, Ye2019, Jaspan2015}. 

There has been significant research attention in recent years on methods that reconstruct sub-sampled MRI data with neural networks \cite{wang2016accelerating, kwon2017parallel, hammernik2018learning, yazdanpanah2019deep, liu2020rare, yang2016deep, yang2018admm,  zhang2018ista,   zhu2018image, quan2018compressed, mardani2018deep, aggarwal2018modl,  Ahmad2020, wang2022dimension, chen2022ai}. The majority of these works use fully-supervised training. To train a network in a fully-supervised manner,  there must be a dataset comprised of fully sampled k-space data $y_{0,t} \in \mathds{C}^{N }$, where  $N$ is the dimension of k-space multiplied by the number of coils,  and paired sub-sampled data $y_t = M_{\Omega_t} y_{0,t}$. Here, $t$ indexes the training set and $M_{\Omega_t} \in \mathds{R}^{N \times N}$ is a sub-sampling mask with sampling set $\Omega_t$, so that the $j$th diagonal of $M_{\Omega_t}$ is 1 if $j \in {\Omega_t}$ and zero otherwise. Then a network $f_\theta$ with parameters $\theta$ is trained  by seeking a  minimum of a non-convex loss function:
\begin{align}
    \hat{\theta}  = \underset{\theta}{\arg \min} \sum_t L(f_\theta (y_t), y_{0,t}),  \label{eqn:fully_sampled_loss}
\end{align}
which could be, for example, an $\ell_p$ norm in the image domain after coil combination \cite{zbontar2018fastmri}. The network $f_{\tht}$ estimates the ground truth in the image domain or k-space depending on the choice of loss function.  For a k-space to k-space network, $y_{0,s}$ can be estimated with $\hat{y}_s   = f_{\tht} (y_s)$, where $s$ indexes the test set. 


Given sufficient representative training data, fully-supervised networks can yield substantial reconstruction quality gains over sparsity-based compressed sensing methods. There are a number of large datasets available for fully supervised training, such as the fastMRI knee and brain data \cite{zbontar2018fastmri}. However, for  many contrasts, orientations, or anatomical regions of interest,  fully sampled datasets are not publicly  available. Fully sampled data is rarely acquired as part of a normal scanning protocol, so acquiring sufficient training data for a specific application is highly resource intensive. In some cases, it may not even be technically feasible to acquire such data \cite{uecker2010real, haji2018validation, lim20193d}. Therefore, for MRI reconstruction with deep learning to be applicable to datasets acquired using only standard protocols, a training method that uses solely sub-sampled data is required. 

There have been several attempts to train networks with only sub-sampled MRI data \cite{9442767, tamir2019unsupervised, huang2019deep, cole2020unsupervised, yaman2020self, liu2021mri, aggarwal2021ensure, zeng2021review, hu2021self}, some of which are based on methods from the denoising literature \cite{lehtinen2018noise2noise, krull2019noise2void,  batson2019noise2self, moran2020noisier2noise, xie2020noise2same, hendriksen2020noise2inverse, kim2021noise2score}. One such approach is Noise2Noise  \cite{lehtinen2018noise2noise}. Rather than mapping $y_t$ to $y_{0,t}$, Noise2Noise trains a network to map $y_t$ to another sub-sampled k-space $y_{T} = M_{\Omega_{T}} y_{0,T}$ where ${\Omega_{T}}$ and ${\Omega_{t}}$ are independent and $y_{0,T} = y_{0,t}$ when $t = T$ \cite{huang2019deep}. A limitation of Noise2Noise is that it requires paired data, so the dataset must contain two independently sampled scans of the same k-space \cite{liu2020rare}, which is not part of standard protocols. Further, unless compensated for \cite{gan2022deformation}, any motion and phase drifts between scans would cause the paired data to be inconsistent, violating the central assumption that underlies the method.


SSDU \cite{yaman2020self} is a recently proposed method for ground-truth free training that does not require paired data. SSDU partitions the sampling set $\Omega_t$ into two disjoint sets: $\Omega_t = A_t \cup B_t$, where  $A_t \cap B_t = \emptyset$. Then the network is trained to recover $M_{A_t} y_t$ from $M_{B_t} y_t$:
\begin{align}
    \tht  = \underset{\theta}{\arg \min} \sum_t L(M_{A_t} f_\theta (M_{B_t} y_t), M_{A_t} y_{t}).  \label{eqn:ssdu_loss}
\end{align}
At inference, the estimate $f_{\tht} (y_s)$ is used. With a physics-guided network architecture, SSDU was found to have a reconstruction quality comparable with fully supervised training given certain empirically selected choices of $A_t$ and $B_t$.  However, it was presented without theoretical justification. Although SSDU has similarities with Noise2Self \cite{batson2019noise2self}, Noise2Self's analysis has a strong requirement on independent noise, so do not apply to k-space sampling in general. 

\subsection{Contributions}

This paper considers the recently proposed Noisier2Noise framework \cite{moran2020noisier2noise}, which was originally constructed for denoising problems. We modify Noisier2Noise so that it can be applied to variable density sub-sampled MRI data. To our knowledge, this is the first work that applies Noisier2Noise to image reconstruction. Like SSDU, the proposed modification of Noisier2Noise does not require paired data, and involves training a network to map from one subset of $\Omega_t$ to another. While SSDU recovers one disjoint set from the other, Noisier2Noise applies a second sub-sampling mask to the data, $\widetilde{y}_t = M_{\Lambda_t} y_{t} =  M_{\Lambda_t} M_{\Omega_t} y_{0,t}$, and the network is trained to recover $y_{t}$ from $\widetilde{y}_t$ with an $\ell_2$ loss. Then, at inference, the fully sampled data is estimated via a correction term based on the distributions of $\Lambda_t$ and $\Omega_t$ that ensures that the estimate is correct in expectation.

Despite their superficial differences, we show that, in fact, SSDU and Noisier2Noise are closely related. Specifically, we demonstrate that SSDU is a version of Noisier2Noise with a particular loss function modification that removes the need for the correction term at inference. The primary contribution of this paper is the use of Noisier2Noise to theoretically explain SSDU's excellent empirical performance. Specifically, we show that SSDU with an $\ell_2$ loss correctly estimates fully sampled k-space in expectation: see Section \ref{sec:SSDU}.

The second contribution of this paper is the proposal of two modifications of SSDU that significantly improve its reconstruction quality and robustness to the parameters of $M_{\Lambda_t}$, both of which arise as a consequence of SSDU's connection to Noisier2Noise. Firstly, we use Noisier2Noise to inform SSDU's sampling set partition: we show that SSDU's performance improves when $B_t$ has the same type of distribution as the original mask $\Omega_t$, but not necessarily with the same parameters. Secondly, we show that SSDU's performance improves when a particular weighting is employed in the loss function. This non-trivial weighting, which arises as a consequence of the novel theoretical analysis of SSDU, depends on the distributions of $\Lambda_t$ and $\Omega_t$ and has minimal additional computational cost: see Section \ref{sec:KW-ssdu}.

Although this paper focuses on MRI reconstruction, we emphasize that none of the theoretical developments are specific to k-space. This framework is therefore applicable to any image reconstruction problem with a forward model that involves random sub-sampling, such as low dose x-ray computed tomography \cite{kang2017deep} or astronomical imaging \cite{8081654}. 

\section{Theory}

This section describes how the Noisier2Noise framework can be applied to sub-sampled data. Additive and multiplicative noise versions of Noisier2Noise are proposed in \cite{moran2020noisier2noise}.  Based on the observation that a k-space sub-sampling mask can be considered as multiplicative ``noise", we extend Noisier2Noise to image reconstruction by modifying the latter. It is standard practice in MRI to sub-sample k-space with variable density, so that low frequencies, where the spectral density is larger, are sampled with higher probability \cite{Lustig2007}. Since the multiplicative noise version of standard Noisier2Noise assumes uniformity, this requires a modification of the framework to variable density sampling.

\subsection{Variable density Noisier2Noise for reconstruction \label{sec:vdn2n}}

The terms in the measurement model $y_t = M_{\Omega_t}y_{0,t}$ can be considered as instances of random variables. We denote $Y = \Mo Y_0$, where  $Y$, $\Mo$ and $Y_0$ are the random variables  corresponding to  $y_t$,  $M_{\Omega_t}$, and  $y_{0,t}$ respectively. Now consider the multiplication of $Y$ by a second mask represented by the random variable $\Mt$, 
\begin{align*}
	\Yt = \Mt Y = \Mt \Mo Y_0,
\end{align*}
so that $\Yt$ is  a further sub-sampled random variable. The following result states how the expectation of $Y_0$ can be computed from $\Yt$ and $Y$. Here, and throughout this paper, $\Eds [\cdot]$ is used to denote the expectation over all random variables within the brackets.
\begin{claim}
\label{clm:n2n}
When $\Eds [M_{\Omega, jj}] = p_j   > 0$ and $\Eds [M_{\Lambda,jj}] = \pt_j  <1$ for all $j$, the expectation of $Y_0$ given $\Yt$ is
\begin{align}
	\mathds{E}[ Y_0 | \Yt ] = (\mathds{1} - K)^{-1} (\mathds{E}[Y | \widetilde{Y}] - K \Yt),  \label{eqn:Eyyt}
\end{align}
where $K$ is a diagonal matrix defined as
\begin{align}
	K =  (\mathds{1} - \Pt P)^{-1}(\mathds{1} - P) \label{eqn:Kdef}
\end{align}
for $P = \Eds[\Mo]$ and $\Pt = \Eds[\Mt]$.
\end{claim} 
\begin{proof}
See Appendix \ref{app:expec_proof}, which is based on the proof given in Section 3.4 of \cite{moran2020noisier2noise}.
\end{proof}

Eqn. \eqref{eqn:Eyyt} generalizes the version of Noisier2Noise proposed  for uniform, multiplicative noise in \cite{moran2020noisier2noise} to variable density sampling. The difference between the uniform and variable density versions is the matrix $K$, which is a scalar in \cite{moran2020noisier2noise}. For the special case where $\Mo$ and $\Mt$ are uniformly random sub-sampling masks, $P$, $\Pt$ and therefore $K$ are proportional to the identity matrix, and \eqref{eqn:Eyyt} simplifies to the uniform version. The mathematical requirement that $p_j > 0$ and $\pt_j <1$ for all $j$ simply ensures that $(\mathds{1} - K)$ is invertible: see Appendix A.

Eqn. \eqref{eqn:Eyyt} implies that $\mathds{E}[ Y_0 | \Yt ]$ can be estimated without fully sampled data by training a network to estimate $\mathds{E}[ Y | \Yt]$. To do this, a network can be trained to minimize 
\begin{align}
	\theta^*  = \underset{\theta}{\arg \min} \hspace{0.1cm} \Eds[ \| W (f_{\theta} (\Yt) - Y) \|_2^2 | \Yt ] \label{eqn:weighted_l2_rv}
\end{align} 
for a full-rank matrix $W$. The minimum occurs when the gradient with respect to $\theta$ is zero:
\begin{multline*}
\nabla_\theta \Eds[ \| W (f_{\theta} (\Yt) - Y) \|_2^2 | \Yt ] \\ = \Eds[  2J W^H W (f_{\theta} (\Yt) - Y)  | \Yt ] = 0,
\end{multline*}
where $J$ is the Jacobian matrix with entries $J_{ij} = \partial f_{\theta} (\Yt)_j / \partial \theta_i$. The number of parameters is typically much greater than $N$, so $J$ has far more rows than columns. Assuming that the rows of $J$ are maximally linearly independent, so the row space is $N$-dimensional, the only solution is
\begin{align}
\Eds [ W^HW( f_{\theta}(\Yt) - Y)  | \Yt ] &= 0. \label{eqn:whw_der}
\end{align}
 If $W$ is full-rank,  $W^HW$ is also full rank, so left-multiplying by $(W^HW)^{-1}$ and using $\Eds [f_{\theta} (\Yt)  | \Yt ] = f_{\theta} (\Yt)$,
\begin{align*}
f_{\theta} (\Yt)  = \Eds [Y | \Yt].
\end{align*}
Therefore, by \eqref{eqn:Eyyt}, a candidate for estimating fully sampled k-space with sub-sampled data only is
\begin{align*}
	\mathds{E}[ Y_0 | \Yt ] = (\mathds{1} - K)^{-1} (f_{\theta^*} (\Yt) - K \Yt).
\end{align*}
This expression does not use $Y$, so does not use all available data. Two candidate approaches for using all available data at inference are considered in this paper. Firstly, one can overwrite known entries of the network output with $Y$:
\begin{align*}
\hat{Y}^{dc} &= (\mathds{1} - \Mo)\mathds{E}[ Y_0 | \Yt ] + Y  \\
& = (\mathds{1} - \Mo)(\mathds{1} - K)^{-1} (f_{\theta^*} (\Yt) - K \Yt) + Y \\
& = (\mathds{1} - \Mo)(\mathds{1} - K)^{-1} f_{\theta^*} (\Yt) + Y,
\end{align*}
where the final step uses $(\mathds{1} - \Mo)\Yt = (\mathds{1} - \Mo)\Mt \Mo Y_0 =0$. Here, the superscript refers to ``data consistent", since the estimate is exactly consistent with $Y$. We emphasize that $\hat{Y}^{dc}$ is consistent with \textit{all} available data $Y$, not just the data in $\Yt$.  Alternatively, similar to the approaches suggested in both SSDU \cite{yaman2020self} and the additive noise examples in Noisier2Noise \cite{moran2020noisier2noise}, one can use singly sub-sampled k-space $Y$ as the network input at inference:
\begin{align}
	\hat{Y} = (\mathds{1} - K)^{-1} (f_{\theta^*} (Y) - K Y) \label{eqn:Yhat_n2n}
\end{align}
Since Claim \ref{clm:n2n} applies to $f_{\theta^*} (\Yt)$, not $f_{\theta^*} (Y)$, \eqref{eqn:Yhat_n2n} is not guaranteed to be correct in expectation. However, it has the advantage that all available data is used by the network. Hence, despite deviating from strict theory, we have found that it performs well in practice: see Section \ref{sec:results}.

This suggests the following procedure, illustrated in Fig. \ref{fig:flow}, for training a network without fully-sampled data. For each sub-sampled k-space in the training set $y_t = M_{\Omega_t} y_{0,t}$, generate a further sub-sampled k-space $\widetilde{y}_t = M_{\Lambda_t} y_{t} =  M_{\Lambda_t} M_{\Omega_t} y_{0,t}$, where $M_{\Lambda_t}$ is an instance of $\Mt$. Then, approximate \eqref{eqn:weighted_l2_rv} by training a  network to minimize the loss function
\begin{align}
    \tht  = \underset{\theta}{\arg \min} \sum_t \| W (f_{\theta} (\widetilde{y}_t) - y_t) \|_2^2, \label{eqn:weighted_l2}
\end{align}
for some full-rank matrix $W$. During inference, estimate fully-sampled k-space with either
\begin{align}
\hat{y}_s^{dc} = (\mathds{1} - M_{\Omega_s})(\mathds{1} - K)^{-1} f_{\tht} (\widetilde{y}_s) + y_s \label{eqn:prop_inf}
\end{align}
or 
\begin{align}
\hat{y}_s = (\mathds{1} - K)^{-1} (f_{\tht} (y_s) - K y_s), \label{eqn:N2Nyhat}
\end{align}
where $s$ indexes the test set.

\begin{figure}
\centering
    \includegraphics[width=\columnwidth]{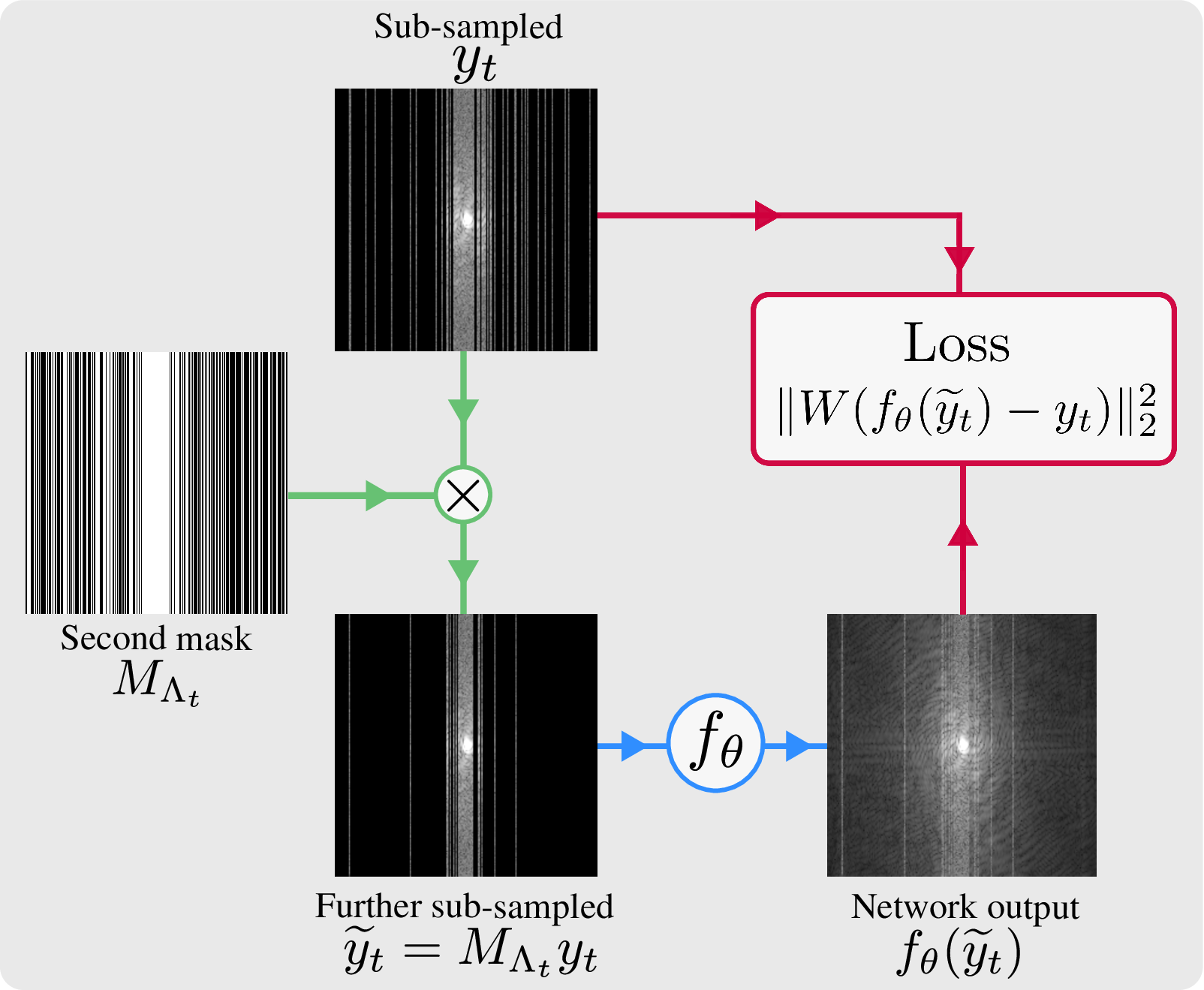}
\caption{A schematic of the self-supervised training methods in this paper. If the loss weighting $W$ is full rank, the training method is variable density Noisier2Noise, as proposed in Section \ref{sec:vdn2n}, whereas if $W = (\eye -  M_{\Lambda_t}) M_{\Omega_t}$ the training method is SSDU: see Section \ref{sec:SSDU} \label{fig:flow}}
\end{figure}


In other words, we train a network to estimate the ``singly" sub-sampled k-space $y_t$ from ``doubly" sub-sampled k-space $\widetilde{y}_t$ and then, during inference, apply a correction based on the diagonal matrix $K$ to estimate the fully sampled data. The correction term only needs to be applied during inference and has minimal computational cost.  

In \cite{moran2020noisier2noise}, only the version with $W = \eye$ was presented. Here we present a version with non-trivial $W$ because it provides a theoretical link to SSDU; Section \ref{sec:SSDU} shows that Noisier2Noise with the rank-deficient $W = (\eye - \Mt)\Mo$ is SSDU exactly. 

Noisier2Noise and SSDU work because the network cannot deduce from $\widetilde{y}_t$  which entries of $y_t$ are non-zero \cite{moran2020noisier2noise}. Therefore, the loss is minimized when the network learns to recover \textit{all} of k-space: see Section \ref{sec:discussion} for a detailed discussion. 

\subsection{Choice of mask distributions}

The only condition on the first mask $\Mo$ from Claim \ref{clm:n2n} is that $p_j > 0$ for all $j$. In other words, the guarantee only applies when there is a non-zero probability that there are sampled examples of all k-space locations in the training set. 

Claim \ref{clm:n2n}  also states that the second mask $\Mt$ must obey $\pt_j <1$ for all $j$. This ensures that there is a non-zero probability that any entry of $\Yt$ is masked. Unlike $\Mo$, whose distribution is determined by the acquisition protocol, the $\Mt$ is chosen freely during training. Following \cite{moran2020noisier2noise}, we suggest using a distribution of $\Mt$ that is the same type as $\Mo$, but not necessarily with the same parameters. For instance, if $\Mo$ is column-wise sampling with variable density, such as in Fig. \ref{fig:flow}, an appropriate $\Mt$ is one that is also column-wise, but possibly with a different variable density distribution.



\subsection{Choice of network}

Noisier2Noise is agnostic to the network architecture. We have found that using the data consistent function 
\begin{align}
     f_\theta (\widetilde{y}_t) = (\mathds{1} - M_{\Lambda_t} M_{\Omega_t}) g_\theta (\widetilde{y}_t) + \widetilde{y}_t, \label{eqn:extra_dc}
\end{align}
where $g_\theta (\widetilde{y}_t)$ is a network with arbitrary architecture, may improve the performance of Noisier2Noise. This is because the $g_\theta (\widetilde{y}_t)$ in \eqref{eqn:extra_dc} only recovers regions of k-space that are not already sampled in $\widetilde{y}_t$, so the network does not need to learn to map sampled k-space locations to themselves. We emphasize that \eqref{eqn:extra_dc} ensures that $f_\theta (\widetilde{y}_t)$ is consistent with $\widetilde{y}_t$, while \eqref{eqn:prop_inf} ensures the estimate $\hat{y}_s^{dc}$ is consistent with $y_s$, which is only applied at inference and cannot be  part of the network architecture when $\widetilde{y}_s$ is used as the input.

Many popular network architectures for MRI reconstruction are based on a sequence of ``unrolled" iterations of a optimization algorithm \cite{hammernik2022physics} such as the Iterative Shrinkage Thresholding Algorithm (ISTA) \cite{Daubechies2004} or the Alternating Direction Method of Multipliers (ADMM) \cite{Boyd2010}. These are variously known as ``physics-guided", ``physics-based" or ``model-based" methods due to their explicit use of the MRI forward model. These architectures typically alternate between a module that recovers missing k-space entries by removing aliasing in the image domain and a module that ensures consistency with the k-space data. This implies that \eqref{eqn:extra_dc}, or possibly a ``soft" version of it where the data is not forced to be exactly consistent, may already be implemented as part of the network architecture. In the experimental evaluation of the methods in this paper we used the Variational Network (VarNet) \cite{hammernik2018learning, sriram2020end}, which is one such architecture where \eqref{eqn:extra_dc} is not necessary. However, in preliminary studies not presented in this paper we found that a U-net \cite{ronneberger2015u}, which does not already employ data consistency, benefited considerably from \eqref{eqn:extra_dc}.

\subsection{Relationship to SSDU \label{sec:SSDU}}

This section shows that SSDU \cite{yaman2020self} with an $\ell_2$ loss is a version of Noisier2Noise with a particular rank-deficient loss weighting matrix $W$. 

To see the connection between SSDU and Noisier2Noise, it is instructive to see the relationship between Noisier2Noise's $\Lambda_t$ and SSDU's  disjoint subsets $A_t$ and $B_t$. Disjoint subsets of $\Omega_t$ can be formed in terms of $\Omega_t$ and $\Lambda_t$ by setting $A_t = \Omega_t \setminus \Lambda_t $ and $B_t = \Omega_t \cap \Lambda_t $. The distribution of $A_t$ and $B_t$ are defined by the distributions of $\Omega_t$ and $\Lambda_t$ and always satisfy $A_t \cup B_t = \Omega_t$ and  $A_t \cap B_t = \emptyset$ as required. In terms of sampling masks, this is written as $M_{A_t} =   (\eye -  M_{\Lambda_t}) M_{\Omega_t} $ and $M_{B_t} = M_{\Lambda_t} M_{\Omega_t}$. Therefore, SSDU's loss \eqref{eqn:ssdu_loss} with a squared $\ell_2$ norm is 
\begin{multline*}
\sum_t \| M_{A_t} f_\theta (M_{B_t} y_t) - M_{A_t} y_{t} \|^2_2 \\ = \sum_t \|  (\eye - M_{\Lambda_t})M_{\Omega_t}(f_{\theta} (\yt_t) - y_t) \|^2_2, 
\end{multline*}
so is exactly  Noisier2Noise with $W = (\eye - M_{\Lambda_t})M_{\Omega_t}$. In other words, while Noiser2Noise's loss is computed over all k-space, SSDU's loss is computed only on indices that are in $\Omega_t$ but not in $\Lambda_t$. 

SSDU's weighting ensures that any indices not sampled in $Y$ are ignored in the loss. One might think that the correct choice for this goal would be $W = M_{\Omega_t}$. However, if a data consistent network is employed, as in \eqref{eqn:extra_dc}, the contribution to the loss from indices in both ${\Omega_t}$ and ${\Lambda_t}$ would be zero because they are consistent by construction. Therefore the loss for $W = M_{\Omega_t}$ and $W = (\eye - M_{\Lambda_t})M_{\Omega_t}$ would be identical. A similar idea was presented for fully supervised learning in \cite{yaman2021improved}, where a mask is applied to the training data multiple times. 

\subsection{Proof of SSDU}

This section shows that SSDU's loss weighting causes the correction $(\eye -K)^{-1}$ at inference to no longer be necessary.  When the weighting matrix $W$ is the random variable $(\eye - \Mt)\Mo$, the network parameters are trained to seek a minimum of
\begin{align}
	\theta^*  = \underset{\theta}{\arg \min} \hspace{0.1cm} \Eds[ \| (\eye - \Mt)\Mo (f_{\theta} (\Yt) - Y) \|_2^2 | \Yt ]. \label{eqn:LSSDU}
\end{align} 
Unlike Noisier2Noise, $W = (\eye - \Mt)\Mo$ is not full-rank, so $f_{\theta^*} (\Yt)  \neq \Eds [ Y | \Yt ]$. The usual theoretical goal for self-supervised methods is to prove that the network is correct in expectation \cite{lehtinen2018noise2noise, krull2019noise2void,  batson2019noise2self, moran2020noisier2noise, xie2020noise2same, hendriksen2020noise2inverse, kim2021noise2score}, as in Claim \ref{clm:n2n} for variable density Noisier2Noise. In the following we state, to our knowledge, the first similar result for SSDU.


\begin{claim}
\label{clm:ssdu}
A network with parameters that minimizes \eqref{eqn:LSSDU} satisfies
\begin{align}
(\eye-K)(\eye - \Mt \Mo)(f_{\theta^*}(\Yt) - \Eds [Y_0 | \Yt])  = 0. \label{eqn:ssdu_clm}
\end{align}
\end{claim}

\begin{proof}
See Appendix \ref{app:ssdu}.
\end{proof}

If $\eye-K$ is invertible, which holds when $p_j > 0$ and $\pt_j <1$ for all $j$, 
\begin{align*}
(\eye - \Mt \Mo)f_{\theta^*}(\Yt)  = (\eye -\Mt\Mo)\Eds [Y_0 | \Yt].
\end{align*}
Therefore, in general, $f_{\theta^*}(\Yt)$ is correct in expectation, but only in regions of k-space that are not sampled in $\Yt$. This contrasts with the variable density Noisier2Noise method presented in Section \ref{sec:vdn2n}, which is correct in expectation for all k-space indices. However, as described in the following, this apparent shortcoming can easily be circumvented by using all available data at inference.

Similarly to Noisier2Noise's \eqref{eqn:prop_inf} and \eqref{eqn:N2Nyhat}, we consider two options for the k-space estimate at inference, both of which use all available data. Firstly, similarly to \eqref{eqn:prop_inf}, the data consistent estimate
\begin{align}
\hat{Y}^{dc} = (\mathds{1} - \Mo) f_{\theta^*} (\Yt) + Y \label{eqn:SSDU_dc}
\end{align}
can be used, which is correct in expectation everywhere in k-space for any network architecture. Alternatively, the SSDU paper \cite{yaman2020self} suggests using 
\begin{equation}
\hat{Y} = f_{\theta^*} (Y) \label{eqn:yhat_ssdu}
\end{equation}
and a physics-guided network architecture. Like  \eqref{eqn:N2Nyhat} for Noisier2Noise, the network input for \eqref{eqn:yhat_ssdu} is singly  sub-sampled, so Claim \ref{clm:ssdu} does not apply and the estimate is not guaranteed to be correct in expectation. Nonetheless, it has the advantage over \eqref{eqn:SSDU_dc} that it uses all available data in the input to the network. As in \cite{yaman2020self}, we have found that \eqref{eqn:yhat_ssdu} performs well in practice when the network architecture includes a data consistency module: see Section \ref{sec:results}.

We emphasize that unlike Noisier2Noise, SSDU does not require the correction term $(\mathds{1} - K)^{-1}$ at inference. This implies that SSDU is less sensitive to inaccuracies in $f_{\theta^*} (\Yt)$, and we have found that SSDU outperforms Noisier2Noise in general: see Section \ref{sec:results}.

\subsection{K-weighted SSDU \label{sec:KW-ssdu}}

Since we train on a finite number of instances of the random variables $Y$, $\Yt$, $\Omega$ and $\Lambda$, the network parameters we obtain in practice, which we denote $\hat{\theta}$, are an approximation of the ideal $\theta^*$ from \eqref{eqn:LSSDU}. In this case, the right-hand-side of \eqref{eqn:ssdu_clm} is not exactly zero. Rather, 
\begin{align}
(\eye-K)(\eye - \Mt \Mo)(f_{\hat{\theta}}(\Yt) - \Eds [Y_0 | \Yt])  = \mathcal{E}, 
\label{eqn:ssdu_E}
\end{align}
where $\mathcal{E}$ is a vector random variable. The vector $\mathcal{E}$ characterizes the difference between a true expectation and the network's estimate of it, which is non-zero for finite data. In other words, $\mathcal{E}$ is a statistical error due to finite sampling. The difference between the trained network's output and the expectation of interest, $\Eds [Y_0 | \Yt]$, is $(\mathds{1} - K)^{-1}\mathcal{E}$. This implies that the network is more sensitive to errors in k-space locations where $(\mathds{1} - K)^{-1}$ is large. 

To compensate for this,  we propose minimizing the following weighted version of SSDU's loss as an alternative to \eqref{eqn:LSSDU}:
\begin{equation*}
	\underset{\theta}{\arg \min} \hspace{0.1cm} \Eds[ \| (\eye - K)^{-\frac{1}{2}}(\eye - \Mt)\Mo (f_{\theta} (\Yt) - Y) \|_2^2 | \Yt ]. \label{eqn:weightSSDU}
\end{equation*}
Introducing $(\eye - K)^{-\frac{1}{2}}$ in the loss cancels the $\eye - K$ in \eqref{eqn:ssdu_E}, so mitigates the error amplification caused by  $\theta^*$ approximation. We find that this version of SSDU, which we refer to as ``K-weighted SSDU" throughout the remainder of this paper, substantially improves the image restoration quality and robustness to training hyperparameters: see Section \ref{sec:results}.  We chose the power $(\eye - K)^{-\frac{1}{2}}$ because it exactly cancels the $\eye - K$ on the left-hand-side of \eqref{eqn:ssdu_E} when the squared $\ell_2$ loss is used; we also tried power $(\eye - K)^{-1}$ and found that, as expected, it did not perform as well in practice.

\subsection{Understanding the need for correction}
 This section intuitively explains why Noisier2Noise  requires correction at inference but SSDU does not. We can write the weighted loss as 
\begin{multline*}
\| W (f_{\theta} (\Yt) - Y) \|_2^2 \\
 = \|W[\Mo \Mt + (\eye  - \Mt)\Mo + (\eye -\Mo)]( f_{\theta}(\Yt) - Y) \|^2_2,
\end{multline*}
where we have used that the term is square brackets equals the identity matrix. When $f_\theta(\Yt)$ is consistent with $\Yt$, such as in \eqref{eqn:extra_dc},  $\Mo \Mt(f_{\theta}(\Yt) - Y) = 0$. Therefore
\begin{multline}
 \| W (f_{\theta} (\Yt) - Y) \|_2^2 
\\ = \|W(\eye  - \Mt)\Mo (f_{\theta}(\Yt) - Y) \|^2_2  + \|W(\eye  - \Mo)  f_{\theta}(\Yt) \|^2_2, \label{eqn:L}
\end{multline}
where we have used $(\eye  - \Mo)Y = 0$. Eqn. \eqref{eqn:L} is SSDU's loss function \eqref{eqn:LSSDU} plus a contribution from all $j \in \Omega_t^c$.

Intuitively, the second term on the right-hand-side of \eqref{eqn:L} causes the proposed method to underestimate regions of k-space with index $j \in \Omega_j^c$. This underestimation is compensated for with $(\mathds{1} - K)^{-1}$ at inference. For SSDU, where $W = (\eye - \Mt)\Mo$, the second term on the right-hand-side of \eqref{eqn:L} is zero, k-space is not underestimated anywhere, and there is no need for a correction term at inference.  


\section{Experimental method}
\label{sec:exp_meth}

\subsection{Description of data}

We used the multi-coil brain and knee data from the fastMRI dataset \cite{zbontar2018fastmri}, which is comprised of multi-channel raw k-space MRI data. The reference fastMRI test set data is magnitude images only, without fully sampled k-space data. Since we also require phase, we discarded the data allocated for testing and generated our own partition into training, validation and test sets. For the brain data, we only used data that was acquired on $16$ coils,  and used training, validation and test set sizes of 127, 19 and 14 volumes (2020, 302, and 224 slices) respectively. For the knee data, the training, validation and test sets consisted of 166, 19 and 14 volumes (5977, 665, and 493 slices) respectively.  We set the network output to be zero in regions of k-space where the reference data had zero padding. 

\begin{table*}[th]
\centering
\scriptsize
\begin{tabular}{c|c|c|c|c}
\textsc{Name}         & \textsc{Loss weighting }$W$ & $\Mt$ \textsc{distribution}                                             & \textsc{Estimate with} $\yt_s$ \textsc{input}& \textsc{Estimate with} $y_s$ \textsc{input}  \\\hhline{=|=|=|=|=}
Unweighted Noisier2Noise & $\eye$ &    1D column-wise                                           & $(\mathds{1} - M_{\Omega_s})(\eye - K)^{-1} f_{\tht} (\yt_s) + y_s$ &$(\mathds{1} - K)^{-1} (f_{\tht} (y_s) - K y_s)$                                                                                                     \\[0pt] \hline 
 2D partitioned SSDU   & $(\eye - M_{\Lambda_t})M_{\Omega_t}$ & 2D Bernoulli & $(\mathds{1} - M_{\Omega_s}) f_{\tht} (\yt_s) + y_s$                                                                                              &  $f_{\tht} (y_s)$        \\[0pt] \hline 
1D partitioned SSDU   & $(\eye - M_{\Lambda_t})M_{\Omega_t}$ & 1D column-wise & $(\mathds{1} - M_{\Omega_s}) f_{\tht} (\yt_s) + y_s$                                                                                                     &$f_{\tht} (y_s)$                                                                                                      \\[0pt] \hline
K-weighted 1D partitioned SSDU   & $(\eye - K)^{-\frac{1}{2}}(\eye - M_{\Lambda_t})M_{\Omega_t}$ & 1D column-wise & $(\mathds{1} - M_{\Omega_s}) f_{\tht} (\yt_s) + y_s$                                                                                                     &$f_{\tht} (y_s)$                                                                                                      \\[0pt] 
\end{tabular}
\caption{The self-supervised methods evaluated in this paper. Here, and throughout this paper, the subscripts $t$ and $s$ index the training and test sets respectively.  Examples of $M_{\Lambda_t} M_{\Omega_t}$ for 2D Bernoulli and 1D column-wise $M_{\Lambda_t}$ are shown in Fig. \ref{fig:second_mask_types}.  \label{tab:meths}}
\end{table*}

\subsection{Network architecture \label{sec:net_arch}}

For $f_\theta$, we used the variant of the VarNet \cite{ hammernik2018learning} that estimates coil sensitivities on-the-fly \cite{sriram2020end}, which performs competitively on the fastMRI leaderboard and is available as part of the fastMRI package\footnote{\underline{https://github.com/facebookresearch/fastMRI}}. After a coil sensitivity estimation module, VarNet uses multiple repetitions of a module based on gradient descent, which is comprised of a data consistency term in k-space and a prior based on a U-net \cite{ronneberger2015u} that acts in the image domain after an inverse Fourier transform and coil combination. The output of the neural network was in k-space. We used $6$ repetitions of the main module, so that our model had around $1.5 \times 10^7$ parameters. Note that in \cite{zbontar2018fastmri}, the Structural Similarity Index (SSIM)  \cite{Wang2004} was used as the loss, while in this paper we use an $\ell_2$ loss.

The only additional operations SSDU and Noisier2Noise require compared to fully-supervised training are simple entry-wise masks, so all methods had similar memory requirements and training time. We trained for 50 epochs, which took around 17 hours on a GTX 1080 Ti GPU with 11GB of RAM for the brain data. For all methods we used the Adam optimizer \cite{kingma2014adam} with a fixed learning rate of $10^{-3}$. Our PyTorch implementation is publicly available on GitHub\footnote{\underline{https://github.com/charlesmillard/Noisier2Noise\_{for}\_recon}}.

\subsection{Distribution of masks}

So that the distribution of the sampling masks were known exactly, we generated our own masks rather than using those suggested in fastMRI. Unless stated otherwise, the distribution of the first mask $\Mo$ was 1D column-wise. We fully sampled the central 10 columns and sampled the remainder with polynomial variable density. We used polynomial order 8, and scaled the probability density $P$ so that it matched a desired acceleration factor. We ran each method with $R_\Omega \in \{4, 8\}$, where $R_\Omega=N/\sum_j p_j$ is the expected acceleration factor. An example at $R_\Omega=4$ is shown in Fig. \ref{subfig:singly}.

In \cite{moran2020noisier2noise}, it is suggested that the distribution of Noisier2Noise's second random variable is the same as the first, but not necessarily with the same distribution parameters. Therefore, for Noisier2Noise's second mask $\Mt$, we used the same type of distribution as $\Mo$ with a different variable density. An example with $R_\Omega=4$ and $\Rt = N/\sum_j \pt_j = 2$ is shown in Fig. \ref{subfig:doubly_col}. Concretely, we define two masks as having the same `type' of distribution when the conditional dependence of the sampling set indices is the same. Let $p_{j|k} = \Pds [ j \in \Omega | k \in \Omega]$. If $ p_{j|k} = p_j$ for all $j$ and $k$, the entries are independent and the mask is the type `2D Bernoulli'.  If $p_{j|k} = 1$ when $j$ and $k$ are in the same k-space column and $p_{j|k} = p_j$ otherwise, the mask is the type `1D column-wise'. The experiments in this paper focus on these two types of masks; other types are discussed in Section \ref{sec:discussion}. We emphasize that constraining a mask to a type does not constrain the $p_j$s, which define the variable sampling density. 

To ensure that $\pt_j < 1$ everywhere, we set $\pt_j = 1-\epsilon$ in the central 10 columns of k-space, where $\epsilon$ is a small real constant. The network architecture ensures that the central region is consistent  with the input, so $\epsilon$ can be small without penalty. We used $\epsilon = 10^{-3}$. 

In order to be a realistic simulation of prospectively sub-sampled data, the sampling set $\Omega_t$ must be fixed for all epochs. However, $\Lambda_t$ need not be. Therefore, we re-generated $M_{\Lambda_t}$ from the distribution of $\Mt$ once per epoch. Since the network sees more samples from the distribution of $\Mt$, the loss function is closer to \eqref{eqn:weighted_l2_rv}, so $f_{\tht}$ is expected to be a more accurate approximation of $\Eds [Y|\Yt]$. This has similarities with training data augmentation, as each slice is used to generate several inputs to the network \cite{yaman2020multi}.

\subsection{Comparative methods}

We trained Noisier2Noise using different weightings of the $\ell_2$ loss stated in \eqref{eqn:weighted_l2}. For each self-supervised method, we considered two possible estimates at inference: one with the doubly sub-sampled $\yt_s$ as the network input and the other with the singly sub-sampled $y_s$. The methods and their two estimates at inference are summarized in Table \ref{tab:meths}.

We trained with  $W= \mathds{1}$, referred to as ``Unweighted Noisier2Noise". By Claim \ref{clm:n2n}, Unweighted Noisier2Noise requires a $(\eye - K)^{-1}$ correction at inference: see Table \ref{tab:meths}. We have found that the need for correction substantially reduces the image quality compared to SSDU, so do not recommend using Unweighted Noisier2Noise in practice. Nonetheless, we include some  Unweighted Noisier2Noise results to illustrate the value of SSDU's loss weighting.


\begin{figure}
\centering
\subfloat[An example $M_{\Omega_t}$ \label{subfig:singly}]{%
    \includegraphics[width=0.31\columnwidth]{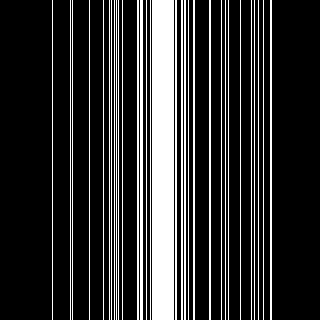}}
    \hspace{0.1cm}
\subfloat[$M_{\Lambda_t} M_{\Omega_t}$ for 1D partitioned SSDU and Noisier2Noise. \label{subfig:doubly_col}]{%
    \includegraphics[width=0.31 \columnwidth]{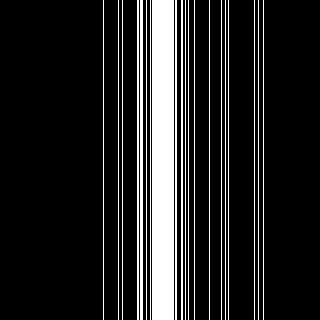}}
\hspace{0.1cm}
\subfloat[$M_{\Lambda_t} M_{\Omega_t}$ for 2D partitioned SSDU \label{subfig:doubly_bern}]{%
    \includegraphics[width=0.31 \columnwidth]{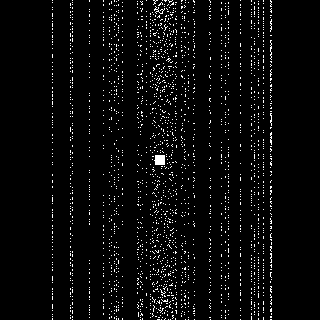}}
\caption{An example of the singly sub-sampled mask $M_{\Omega_t}$, and doubly sub-sampled $M_{\Lambda_t} M_{\Omega_t}$ with two $\Mt$ distribution types. Here, the acceleration factor of the first mask is $R_\Omega=4$ and the second is $\Rt = 2$.\label{fig:second_mask_types}}
\end{figure}

We also trained Noisier2Noise with $W = (\eye-\Mt)\Mo$ which, based on the relationship described in Section \ref{sec:SSDU}, we refer to as ``SSDU", despite some differences between our implementation and \cite{yaman2020self}. In \cite{yaman2020self}, a mixture of an $\ell_1$ and $\ell_2$ loss was used, whereas here, so that it can be directly compared with Unweighted Noisier2Noise, we used an $\ell_2$ loss. We also used a different $\Mo$ distribution, dataset and network architecture to \cite{yaman2020self}.

SSDU \cite{yaman2020self} was originally applied to an architecture that requires pre-computed sensitivity maps.  It was suggested that $M_{B_t}$ has a fully sampled $4 \times 4$ central region and 2D Gaussian variable density otherwise, so that high frequencies are sampled with higher probability. For the architecture considered in this paper, which has a coil sensitivity estimation module, we found that increasing the size of the fully sampled central region  considerably improved the method's performance. Since $M_\Omega$ has 10 fully sampled central columns, we increased the size of the central region of $M_\Lambda$ to $10 \times 10$. 

\begin{figure}
    \includegraphics[width=0.5\textwidth]{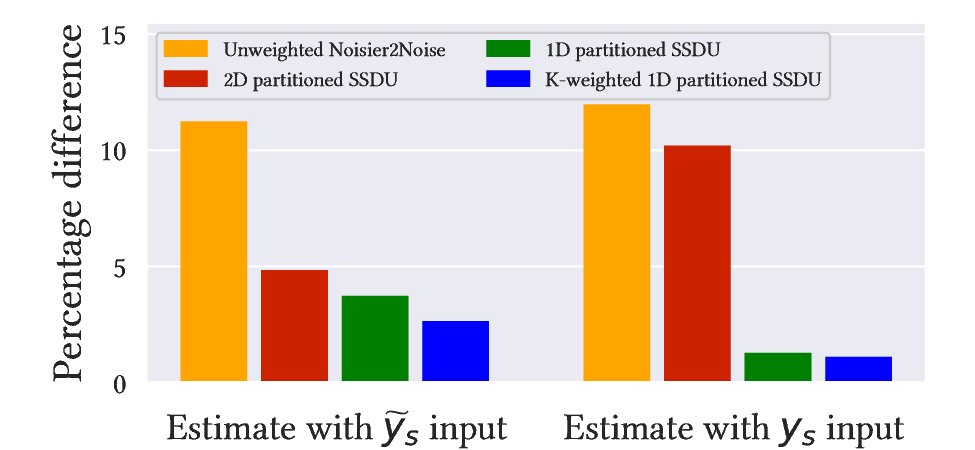}
    \caption{The mean test set NMSE percentage difference between fully supervised and each methods at $R_\Omega=8$ and a 1D distributed $\Mo$, where $\Rt$ has been tuned to minimize the test set NMSE. Fig. S1 shows a similar plot for $R_\Omega=4$. }
    \label{fig:bar_plots}
\end{figure} 

\begin{figure}
    \includegraphics[width=0.5\textwidth]{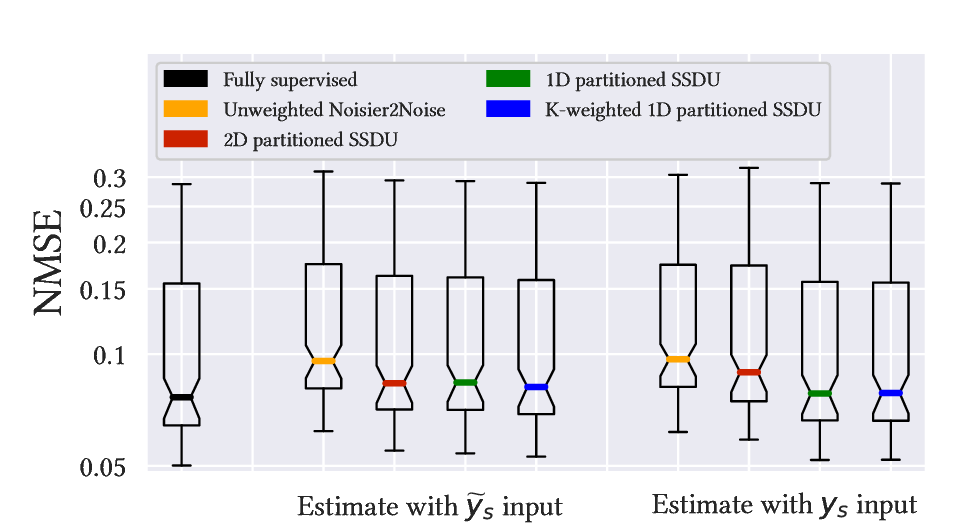}
    \caption{The NMSE for all methods at $R_\Omega=8$ and a 1D distributed $\Mo$, where $\Rt$ has been tuned to minimize the test set NMSE. Fig. S2 shows a similar plot for $R_\Omega=4$ and the exact numerical values are in Table S1. }
    \label{fig:box_plots}
\end{figure} 

\begin{figure}
    \centering
    \includegraphics[width=0.5\textwidth]{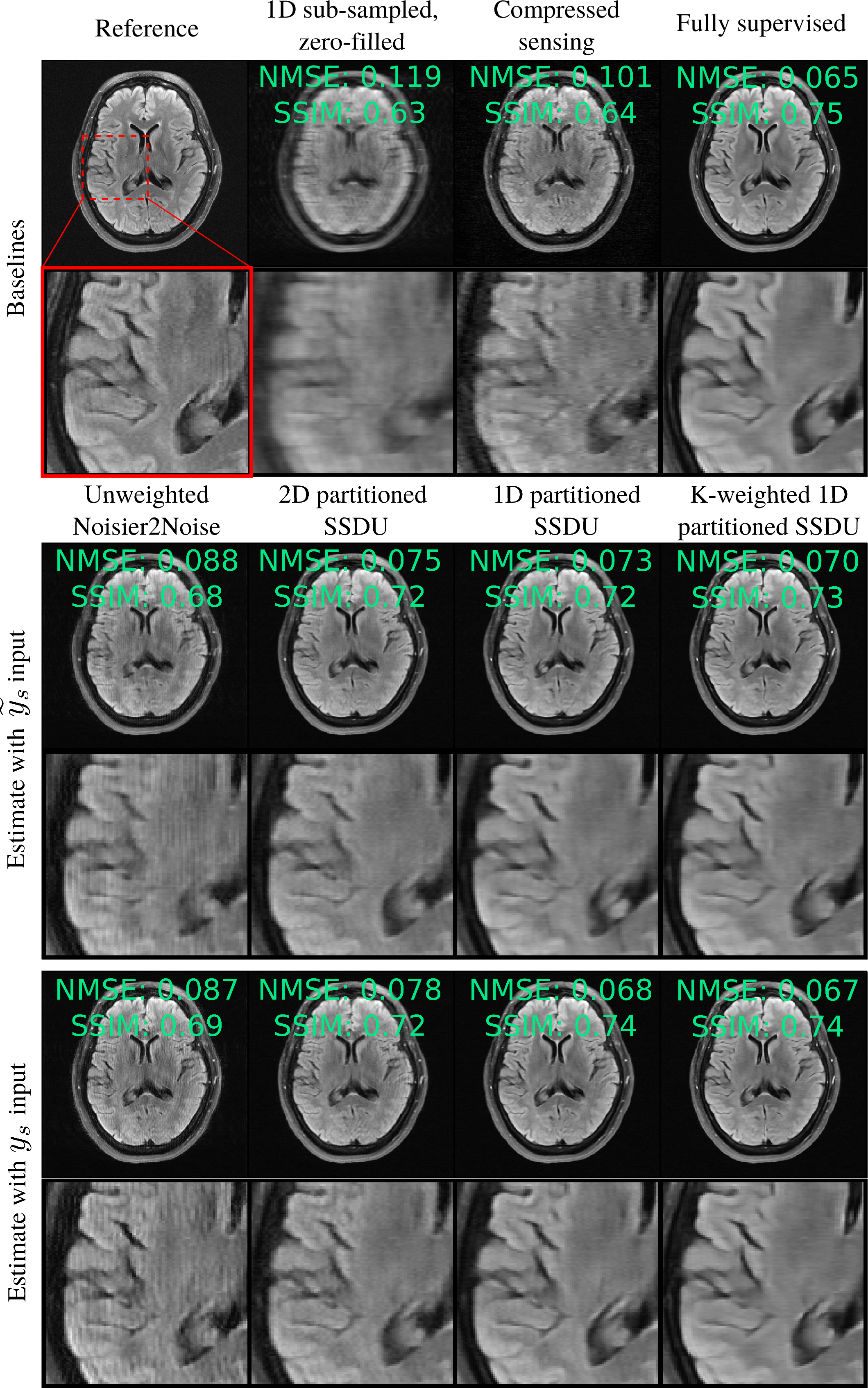}
    \caption{A reconstruction example with a 1D sub-sampled $\Mo$ and $R_\Omega=8$, with a $\Rt$ tuned to minimize the test set NMSE.  A similar figure for $R_\Omega=4$ is in the supplementary material, Fig. S3.}
    \label{fig:8x}
\end{figure}

As the probability of sampling each location in k-space is independent, the sampling set partition proposed in \cite{yaman2020self} is equivalent to a 2D variable density Bernoulli $\Mt$ distribution. To estimate their variable density distribution $\Pt$ we ran the SSDU authors' set partitioning code\footnote{\underline{https://github.com/byaman14/SSDU}} 1000 times on a fully sampled mask and averaged the result. We trained SSDU using a distribution of $\Mt$ of this type, referred to as ``2D partitioned SSDU", illustrated in Fig. \ref{subfig:doubly_bern}. We also trained SSDU using the same distribution type of $\Mt$ as $M_\Omega$, as in Fig. \ref{subfig:doubly_col}. We refer to this method as ``1D partitioned SSDU", or ``K-weighted 1D partitioned SSDU" when a $(\eye - K)^{- \frac{1}{2}}$ weighting is used in the loss as described in Section \ref{sec:KW-ssdu}.  Like Unweighted Noisier2Noise, $M_{\Lambda_t}$ was re-generated once per epoch \cite{yaman2020multi}. We emphasize that although 2D partitioned SSDU has a similar $\Mt$ distribution as in \cite{yaman2020self}, the distribution of $\Mo$ here is random variable density columns, not equidistant columns as in \cite{yaman2020self}. Therefore, 2D partitioned SSDU is  not necessarily expected  to perform as well as SSDU in \cite{yaman2020self}.

As a best-case target, we also trained using a fully supervised method with an (unweighted) $\ell_2$ loss. All deep learning methods had the same network architecture and training hyperparameters, as described in \ref{sec:net_arch}. 

Finally, as a comparative method that does not use deep learning, we ran a compressed sensing algorithm with a sparse model on  wavelet coefficients, which we implemented via the Berkeley Advanced Reconstruction Toolbox (BART) \cite{uecker2016bart}. We used BART's default settings with fourth-order Daubechies wavelets and a sparse weighting of $\lambda = 2 \times 10^{-3}$.

\subsection{Quality metrics}

To evaluate the reconstruction quality, we computed the Normalized Mean Squared Error (NMSE) in k-space on the test set: $\|\hat{y}_s - y_{0,s}\|^2_2/\|y_{0,s}\|^2_2$. We also computed the image-domain root-sum-of-squares (RSS), $\hat{x}_s = (\sum_c | F^H y_{s,c} |^2)^{1/2}$ where $y_{s,c}$ is the k-space entries on coil $c$ and $F$ is the discrete Fourier transform, cropped the RSS estimate to a central $320 \times 320$ region and computed  the SSIM, as suggested in fastMRI \cite{zbontar2018fastmri}.

\begin{figure}
    \centering
    \includegraphics[width=0.48\textwidth]{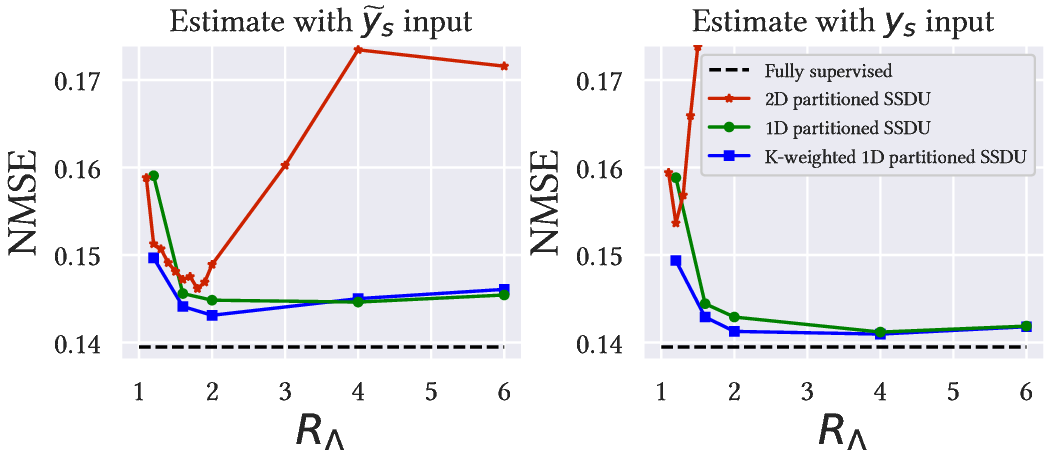}
    \caption{The dependence of the test set NMSE on the acceleration factor of the second mask $\Mt$, denoted as $\Rt$, at $R_\Omega=8$ for both outputs. 1D partitioned SSDU is far more robust to the tuning of $R_\Lambda$ than 2D partitioned SSDU. Fully supervised learning does not use a second mask $\Mt$, so has the same performance for all $\Rt$. A similar figure for $R_\Omega=4$ is in the supplementary material, Fig. S3.}
    \label{fig:line_4x}
\end{figure} 

\begin{figure}[t]
    \centering
    \includegraphics[width=0.5\textwidth]{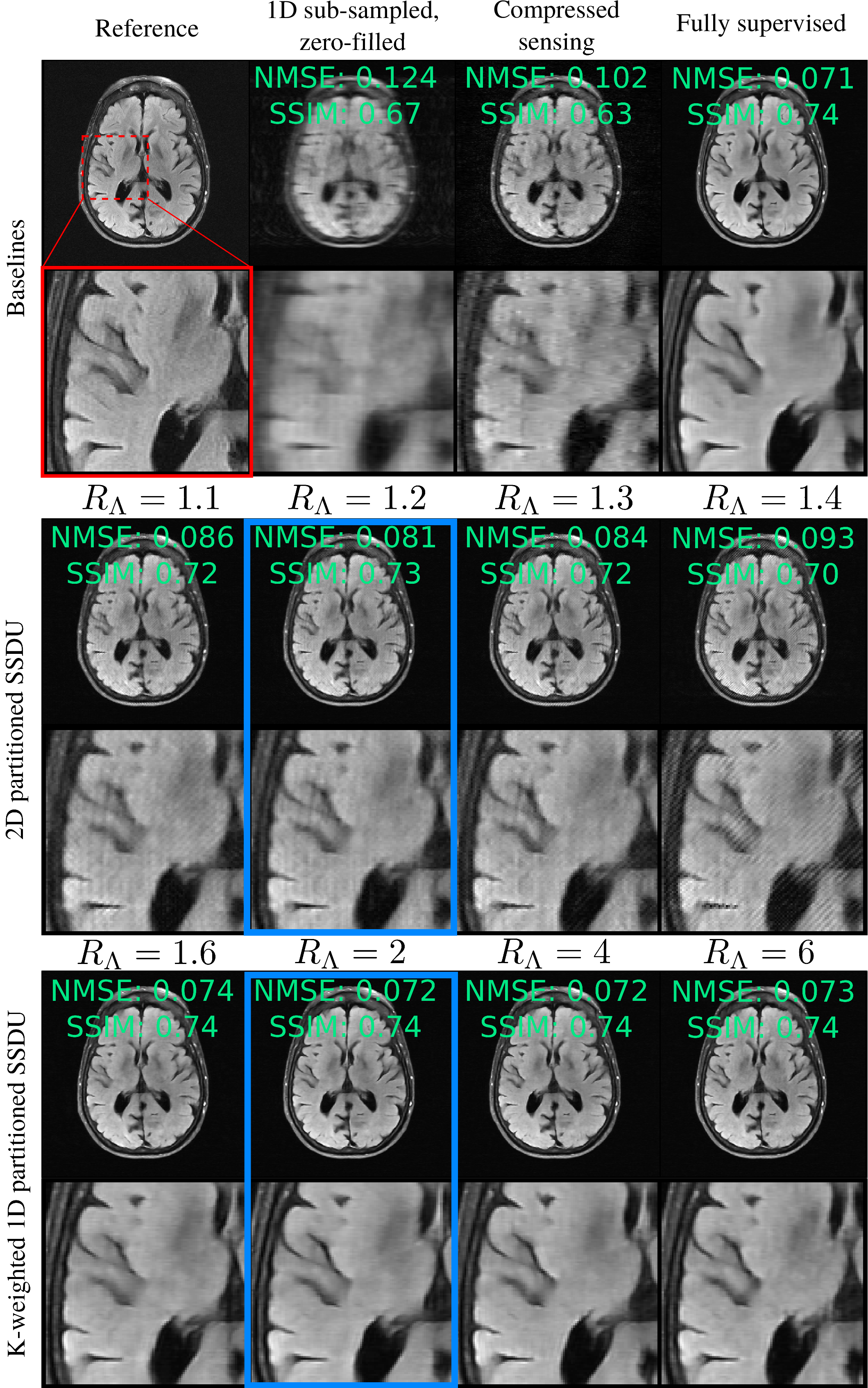}
    \caption{Robustness to $R_\Lambda$, where the blue box highlights the case where $R_\Lambda$ is tuned. K-weighted 1D partitioned SSDU is very robust to $R_\Lambda$, with very similar restoration quality for all $R_\Lambda$ between 1.6 and 6. 2D partitioned SSDU is far more sensitive, with substantial degradation in image quality for mistunings as small as 0.1.  Here, we show the estimate with $y_s$ input only.}
    \label{fig:8x_robustness}
\end{figure}

\begin{figure}
        \includegraphics[width=0.5\textwidth]{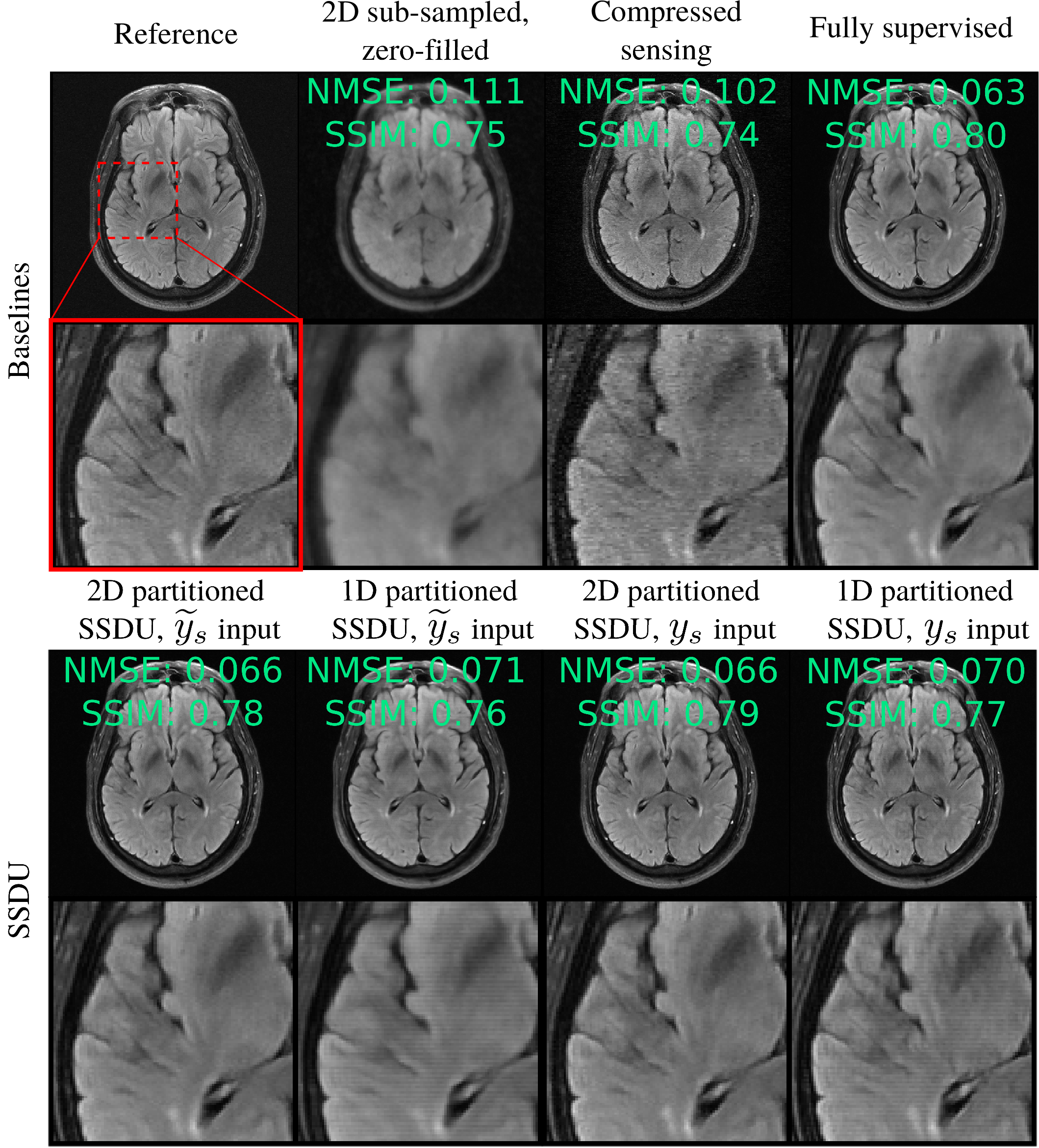}
    \caption{A reconstruction example from the brain fastMRI dataset with a 2D Bernoulli distributed $M_\Omega$ and $R_\Omega=8$. Compared to Fig. \ref{fig:8x}, the comparative performance of the SSDU algorithms are switched: here, 2D partitioned SSDU performs similarly to fully supervised training, while 1D partitioned SSDU suffers from streaking artifacts.}
    \label{fig:8x_bern}
\end{figure}

\begin{figure}[t]
    \centering
        \includegraphics[width=0.5\textwidth]{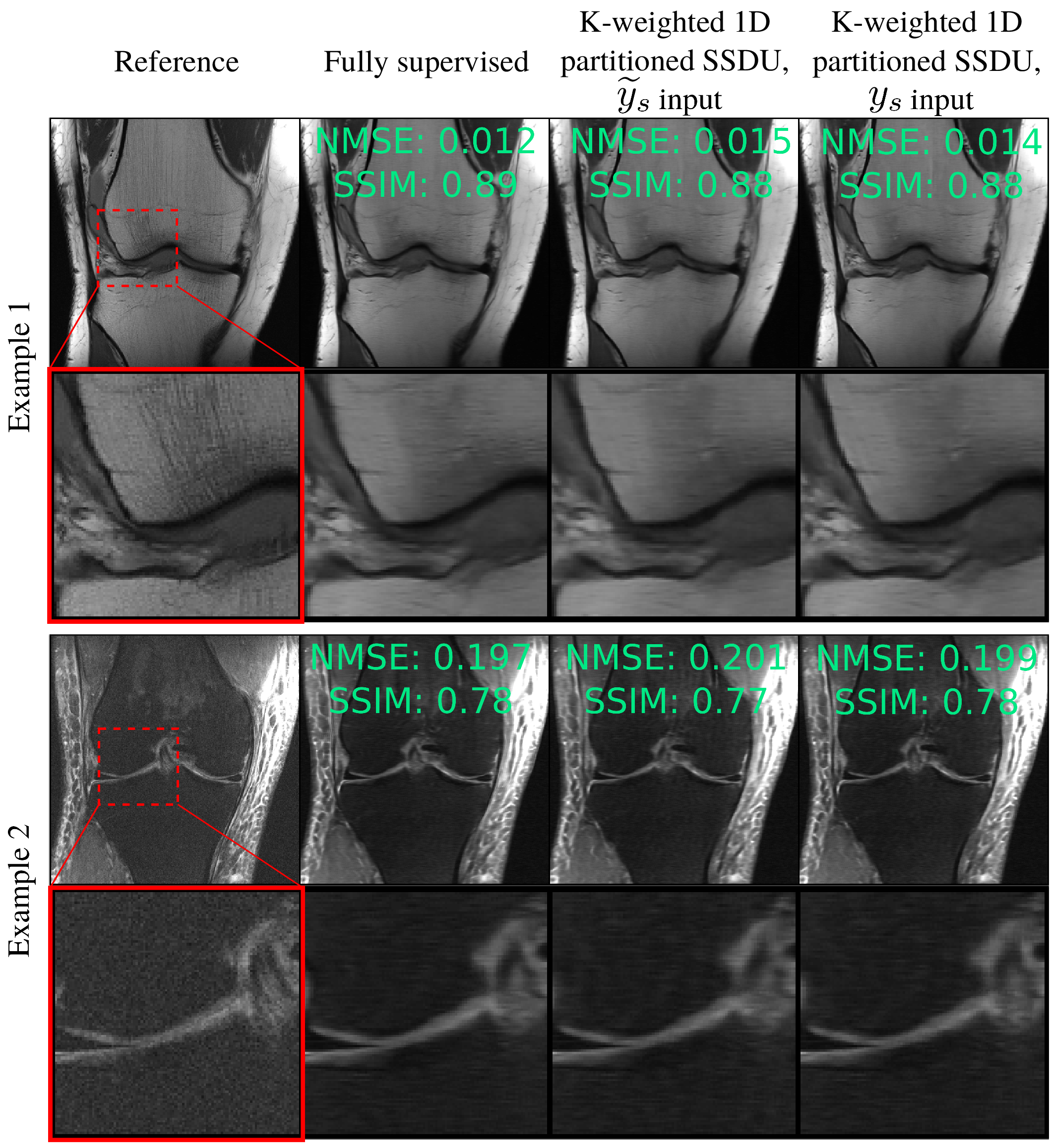}
    \caption{Two reconstruction examples of K-weighted 1D partitioned SSDU from the knee fastMRI dataset, where $\Mo$ is 1D.  As in Fig. \ref{fig:8x}, K-weighted 1D partitioned SSDU's restoration quality is very similar to fully supervised training.}
    \label{fig:8x_knee}
\end{figure}


\section{Results}
\label{sec:results}

For brevity, the results presented here focus on $R_\Omega=8$. Similar results for the brain data at $R_\Omega=4$ are shown in the supplementary material: see figures S1-S4. 

For the brain data, we evaluated the dependence of the methods' performance on the distribution of $\Mt$ by varying the parameters so that the sub-sampling factor $\Rt$ changed. We trained with $\Rt \in \{1.2, 1.6, 2, 4, 6\}$, except for 2D partitioned SSDU, which we found needed finer tuning and a smaller $\Rt$ for the best performance, so we trained with $\Rt \in \{1.1, 1.2, \ldots, 2, 3, 4, 6\}$.

\subsection{Performance with tuned $\Rt$ \label{sec:tunedR}}
This section focuses on the case where $\Rt$ has been tuned to minimize the ground truth test set NMSE. Figures \ref{fig:bar_plots} and S1 show bar charts of the percentage difference between fully supervised training and each method: $(\mu - \mu_{full})/\mu_{full}$ where $\mu$ and $\mu_{full}$ are the mean NMSE of interest and mean NMSE of fully supervised training respectively. The best performance was for K-weighted 1D partitioned SSDU with a $y_s$ input; its mean NMSE was only 1.1\% and 0.8\% larger than fully supervised for $R_\Omega = 8, 4$ respectively. Figures \ref{fig:box_plots} and S2 show box plots of the NMSE of each method for $R_\Omega = 8$ and $R_\Omega=4$ respectively: see Table S1 of the supplementary material for the numerical values. 

To evaluate whether the proposed changes to SSDU were statistically significant, we performed a one-sided Wilcoxon signed-rank test with $p$-value 0.01 on the test set NMSEs.  For both the $y_s$ and $\yt_s$ inputs, we found that there was a significant statistical difference between 2D and 1D partitioned SSDU. We also found that the difference between 1D partitioned SSDU and K-weighted 1D partitioned SSDU was statistically significant. 

Figures \ref{fig:8x} and S3 show RSS estimates from the test set at $R_\Omega=8$ and $R_\Omega=4$ respectively. Qualitatively, K-weighted 1D partitioned SSDU performs the most similarly to fully supervised training. Although 2D partitioned SSDU has a competitive quantitative score for the estimate with $\yt_s$ input, it exhibits some streaking artifacts. 

Unweighted Noisier2Noise's performance was substantially worse than SSDU. Therefore we compare SSDU and its modifications only in the remainder of this paper. 

\subsection{Robustness to $\Rt$} 

For actual, prospectively sampled data, it would not be possible to tune $R_\Lambda$ on the ground truth test set NMSE. The practicality of SSDU therefore depends greatly on the robustness to $R_\Lambda$. Figures \ref{fig:line_4x} and S4 show the dependence of the mean test set NMSE on $\Rt$ for $R_\Omega = 8$ and $R_\Omega = 4$ respectively. K-weighted 1D partitioned SSDU was the most robust to the tuning of $\Rt$. 2D partitioned SSDU was the least robust, especially for the estimate with $y_s$ input. This is visualized in Fig. \ref{fig:8x_robustness}, which shows reconstruction examples for a number of $R_\Lambda$s. K-weighted 1D partitioned SSDU performs very similarly for all $R_\Lambda$s between 1.6 and 6, while 2D partitioned SSDU's restoration quality significantly degrades qualitatively and quantitatively for mistunings as small as $0.1$.

\subsection{Performance on 2D sampled brain data}

To further evaluate the role of the partitioning distribution, we also ran 1D and 2D partitioned SSDU on the brain data with a 2D Bernoulli sampled $\Mo$. In this case, the type matching of the second mask to $\Mo$ is switched: 2D partitioned SSDU's second mask has the same type of distribution as the first, while 1D partitioned SSDU has a different type. For $\Mo$, we used a fully sampled $10\times 10$ central region and a polynomial variable density that samples low frequencies with higher probability otherwise. We used $\Rt = 1.2$ and $\Rt = 4$ for 2D and 1D partitioned SSDU respectively. All other hyperparameters and network specifics were unchanged. 

In this case, the best performance was 2D partitioned SSDU, which performed very similarly to fully supervised training: see Fig. \ref{fig:8x_bern}. The $\yt_s$ input had a mean test set NMSE of 0.141 and 0.144 for 2D and 1D partitioned SSDU respectively, and the $y_s$ input had 0.141 and 0.145, compared with 0.139 for fully supervised training. Although not shown in Fig. \ref{fig:8x_bern} for brevity, we also trained 2D partitioned SSDU with a $(\eye - K)^{-\frac{1}{2}}$ loss weighting. As for 1D partitioned SSDU in Section \ref{sec:tunedR}, we found that this reduced the mean NMSE further to 0.140 for both the $y_s$ and $\yt_s$ input.

\subsection{Performance on 1D sampled knee data}

We also trained K-weighted 1D partitioned SSDU on the fastMRI knee data with the same  network architecture, training hyperparameters, and a 1D distributed $M_\Omega$. The sub-sampling factor of the first and second masks were $R_\Omega=8$ and $\Rt = 2$ respectively. The mean test set NMSE was  0.233 and 0.231 for the estimates with $\yt_s$ and $y_s$ inputs respectively, compared with 0.230 for fully supervised training. Fig. \ref{fig:8x_knee} shows two example reconstructions from the test set, demonstrating competitive performance with fully supervised training qualitatively.

\section{Discussion \label{sec:discussion}}


Due to its need for correction at inference, Unweighted Noisier2Noise had consistently the worst score. We therefore do \textit{not} recommend using Unweighted Noisier2Noise in practice. Rather, we suggest using a variant of SSDU, which has a loss weighting that removes the need for such a correction.   


The hierarchy of 1D and 2D partitioned SSDU depends on the distribution of $\Mo$. In particular, the best performance was when they are both 1D or both 2D. It is conventional wisdom that better reconstruction quality is possible  when k-space is randomly sub-sampled in both spatial dimensions (see, for instance, \cite{deshpande2012optimized}). This is because the image-domain aliasing is incoherent in both dimensions, so is easier to remove. The superior performance of 1D partitioned SSDU compared with 2D partitioned SSDU when $\Mo$ is 1D shows that it is \textit{not} necessarily true that the sampling set partition should also ideally be two-dimensional. Rather, better performance is possible when the distribution of $\Mo$ and $\Mt \Mo$ are of the same type. 


To see why, consider the nature of the aliasing caused by sub-sampling and further sub-sampling k-space, focusing on the example of a random 1D column sampled $\Mo$. Such sampling causes the image-domain aliasing to be horizontally incoherent and vertically coherent. With a 1D column-wise $\Lambda_t$, further horizontal aliasing is introduced. Since the network cannot distinguish between the horizontal aliasing caused by $\Omega_t$ or $\Lambda_t$, the loss is minimized when the aliasing due to \textit{both} is removed. On the other hand, a 2D $\Lambda_t$ introduces some aliasing that is orthogonal to the original aliasing, which is distinguishable in principle. In this case, the loss is minimized when the network removes the aliasing caused by $\Lambda_s$, but not necessarily the original aliasing caused by $\Omega_s$. 
This is visible in figures \ref{fig:8x} and  \ref{fig:8x_bern}, where SSDU fails to completely remove artifacts caused by $\Mo$ when $\Mt$ does not have the same type of distribution.

This implies that, in general, better performance is possible when the distribution of the aliasing of $\yt_t$ and $y_t$ are of the same type.   For both the independent 1D column sampling and 2D Bernoulli sampling considered here, this can be achieved by choosing a $\Mt$ with the same type of distribution as $\Mo$. Recently, in \cite{Blumen2022}, this was also observed empirically for SSDU with random spoke sampling. However, such a procedure does not always achieve this goal. For instance, while the SSDU paper \cite{yaman2020self} considers a fully sampled central region and equidistant column sampling, recovery of images with regular under-sampling is not currently considered in the proposed framework. In this case, a $\Lambda_t$ of the same type would not give a $\yt_t$ with the same aliasing type as $y_t$. The 2D Gaussian variable density partition employed in this paper was originally constructed to handle such sampling patterns, and was found to perform very well in this context. Future work includes establishing the correct sampling set partitions for $\Mo$ distributions not in  \cite{yaman2020self} or covered by the approach suggested here.

We found that K-weighted SSDU further improved the image quality and robustness to $R_\Lambda$. Consider the $j$th entry of the (squared) weighting $(\eye - K)^{-1}$ in terms of sampling probabilities: 
\begin{equation*}
(1 - k_j)^{-1} = \frac{1 - \pt_j p_j}{p_j(1 - \pt_j)} = \frac{\Pds(j \notin \Lambda \cap \Omega) }{\Pds(j \in \Omega \setminus
 \Lambda) }.
\end{equation*}
This leads to the following intuitive interpretation of the proposed loss weighting as compensation for the variable density of $\Omega$ and $\Lambda$. A smaller denominator $\Pds(j \in \Omega \setminus \Lambda)$ implies that the $j$th location occurs less frequently in the loss, which is compensated for by an increased weighting. A smaller numerator $\Pds(j \notin \Lambda \cap \Omega)$ implies that the $j$th location is estimated by the network less frequently, so has a decreased weighting.  

The benefit of the $(\eye - K)^{-1}$ weighting  highlights and addresses a general challenge of self-supervised learning with variable density sampling: regions of k-space sampled with lower probability are underrepresented in the loss. This issue has been noted in other works. For instance, for variable density reconstruction with Noise2Noise, \cite{gan2022self} suggests weighting the loss function by the sampling density.  An alternative approach was suggested in \cite{liu2022iterative}, which suggests passing the training target through the network before it is employed in the loss function.  We note that if the sampling and partitioning had uniform density, such as in \cite{yaman2020multi}, $K$ would also be uniform, so the proposed weighting would not be required. This may explain in part the empirical performance observed in \cite{yaman2020multi}.  


When $\Mo$ was 1D, with the exception of 2D partitioned SSDU, Fig. \ref{fig:line_4x} shows that the estimate with $y_s$ input performed similarly or better than with $\yt_s$  input when $\Rt$ is tuned. This indicates that, for these methods, the advantage of using all the data in the input to the network outweighs the disadvantage that the input data has a different sampling distribution to the training data so is not guaranteed by Claim \ref{clm:n2n} or \ref{clm:ssdu} to be correct in expectation. Heuristically, when $\Mo$ and $\Mt \Mo$ are both variable density column-wise sampled, a network trained on doubly sub-sampled data is likely to also be able to handle singly sub-sampled data.  However, for 2D partitioned SSDU, $\Mt \Mo$ is no longer column-wise, see Fig. \ref{fig:second_mask_types}c. Accordingly, 2D partitioned SSDU was the only method that had a higher NMSE for the $y_s$ input compared to the $\yt_s$ input. 

The best $\Rt$ for 2D partitioned SSDU was lower than competing methods: $\Rt = 1.8$ and $\Rt = 1.2$ for the $y_s$ and $\yt_s$ inputs respectively. In \cite{yaman2020self}, the sampling set partition was quantified in terms of the ratio $\rho = |A_t|/ |B_t|$,  and it was found that $\rho = 0.4$ offered the best performance. Since the $\Mo$ distributions are different here, the optimal $\rho$ is not expected to necessarily be the same. For 2D partitioned SSDU $\Rt = 1.8$ and $\Rt = 1.2$ corresponds to $\rho = 0.52$ and $\rho = 0.21$ respectively, while for the other methods's best performance at $\Rt = 4$ corresponded to $\rho = 0.57$. Therefore the $\rho$ were reasonably similar despite the substantial difference in $\Rt$.

Since the network architecture uses $\yt_t$ in its coil sensitivity estimation module, not $y_t$, it is plausible that the differences between 1D and 2D partitioning could be due to poorer coil sensitivity estimation rather than an intrinsic property of the partition change. To examine this, we re-trained tuned 1D and 2D partitioned SSDU on the 1D sampled brain data with k-space masked to a central $10 \times 10$ region in the coil sensitivity estimation module. We found that the test set NMSE was within 1\% of the usual approach. This verifies that the performance improvement was indeed a consequence of the partition change, not simply a consequence of specifics of the architecture.

Unweighted Noisier2Noise's correction at inference $(\eye - K)^{-1}$ is only valid when an $\ell_2$ loss is used;  we have found that other loss functions do not perform well in practice. This loss leads to smoothing artifacts, even for fully supervised training. For SSDU, since there is no correction term, loss functions other than $\ell_2$ are possible. For instance,  in \cite{yaman2020self}, a mixture of $\ell_2$ and $\ell_1$ was used. Better visual quality may be achievable when SSDU is implemented with a different loss; we do not suggest using an $\ell_2$ loss in general, it is only required here so that it can be compared directly with Noisier2Noise.

For all self-supervised methods in this work, we re-generated $\Lambda_t$ once per epoch. This has similarities to the multi-mask SSDU approach proposed in \cite{yaman2020multi}. However, in \cite{yaman2020multi}, a fixed number $n_\Lambda$  of $\Lambda_t$s were generated for each $\Omega_t$, each of which were treated as an additional member of the training set. Therefore, unlike in this paper, each epoch was $n_\Lambda$ times as long. Future work includes establishing whether it is also advantageous to limit the number of unique $\Lambda_t$s per $\Omega_t$ for the approach considered in this paper.

All methods in this paper were trained without taking measurement noise into account \cite{desai2021noise2recon, chen2022robust}. Recent work by the present authors has shown that the additive and multiplicative versions of Noisier2Noise can be combined to recover higher fidelity images than SSDU in the presence of noise \cite{SimN2N}.

\section{Conclusions and future work}

Based on the observation that SSDU is a version of Noisier2Noise with a particular rank-deficient loss weighting, we proved that SSDU correctly estimates $Y_0$ in expectation.  This analysis led to two proposals that we found significantly improved SSDU's performance in practice. Firstly, we propose employing a distribution of $M_\Lambda \Mo$ that is the same type as the original mask $M_\Omega$. Secondly, we propose introducing a weighting of $(\eye - K)^{-\frac{1}{2}}$ in SSDU's loss. We found that that each of these modifications significantly improved SSDU's test set NMSE and robustness to $\Rt$. 

There are a number of other self-supervised learning methods that also use sampling set partitioning \cite{yaman2020multi, hu2021self, zou2022selfcolearn}, some of which are variants of SSDU. For instance, \cite{hu2021self, zou2022selfcolearn, wang2022parcel} propose training two networks in parallel, one for each sampling subset, with a loss function that includes the difference between the outputs of the two networks. Another recent development is zero-shot SSDU \cite{yaman2021zero}, which shows that sampling set partitioning can also be applied to recover images without a training dataset \cite{ulyanov2018deep}. Future work includes determining whether the theoretical and practical developments of this paper can be extended to these methods.

\section{Appendices}

\subsection{Proof of variable density Noisier2Noise \label{app:expec_proof}}

This appendix proves that when $p_j \neq 0$ and $\pt_j \neq 1$ for all $j$,
\begin{align}
\mathds{E}[ Y_0 | \Yt ] = (\mathds{1} - K)^{-1} (\mathds{E}[Y | \widetilde{Y}] - K \Yt ), \label{eqn:ey0}
\end{align}
where $K =  (\mathds{1} - \Pt P)^{-1}(\mathds{1} - P)$ for $P = \Eds[\Mo]$ and $\Pt = \Eds[\Mt]$.

\begin{proof}
This proof is based on Section 3.4 of Noisier2Noise \cite{moran2020noisier2noise}, but with  more mathematical detail and generalized to variable density sampling.  Following the compressed sensing literature, this paper uses $p_j$ to refer to the probability that the $j$th location in k-space is sampled. This differs to \cite{moran2020noisier2noise}, which uses $p$ to denote the probability that a pixel is \textit{zeroed}. 

We wish to compute $\Eds[Y_j|\Yt_j]$ as a function of $\Eds [Y_{0,j} | \Yt_j]$. To do this, we split $\Eds[Y_j|\Yt_j]$ into two cases, for conditions $\Yt_j \neq 0$ or $\Yt_j = 0$, and subsequently construct an expression that is consistent with both.

\textit{Case 1} ($\Eds[Y_j|\Yt_j \neq 0]$): By the measurement model $\Yt = \Mt Y = \Mt \Mo Y_0$, the singly sub-sampled $Y_j$ must take the same value as $\Yt_j$ when $\Yt_j \neq 0$. Therefore 
\begin{equation}
\Eds[Y_j|\Yt_j \neq 0] = \Yt_j. \label{eqn:case1end}
\end{equation}

\textit{Case 2} ($\Eds[Y_j|\Yt_j = 0]$): Using the partition theorem for expectations, we write $\Eds[Y_j|\Yt_j = 0]$ as the weighted sum of $\Eds [Y_{j} | \Yt_j = 0 \cap Y_j = 0 ]$ and $\Eds [Y_{j} | \Yt_j = 0 \cap Y_j \neq 0 ]$:
\begin{align}
    \Eds[Y_j|\Yt_j = 0] =& \Eds [Y_{j} | \Yt_j = 0 \cap Y_j = 0 ] \cdot k_j \nonumber \\ 
    & + \Eds [Y_{j} | \Yt_j = 0 \cap Y_j \neq 0 ] \cdot (1-k_j), \label{eqn:Eyyt0}
\end{align}
where we define $k_j = \Pds [Y_j = 0 |\Yt_j = 0 ]$. Evaluating each of the terms on the right-hand-side of \eqref{eqn:Eyyt0} in turn: 
\begin{itemize}
\item $\Eds [Y_{j} | \Yt_j = 0 \cap Y_j = 0 ]$: Since the random variable $Y_{j}$ is conditionally zero, its expectation is also zero:
\begin{equation*}
	\Eds [Y_{j} | \Yt_j = 0 \cap Y_j = 0 ] = 0.
\end{equation*}
\item $\Eds [Y_{j} | \Yt_j = 0 \cap Y_j \neq 0 ]$: The measurement model implies that when $Y_j$ is non-zero, and therefore unmasked, it takes the value of  $Y_{0,j}$. Therefore its expectation can be written in terms of the expectation of $Y_{0,j}$:
\begin{equation}
	\Eds [Y_{j} | \Yt_j = 0 \cap Y_j \neq 0 ] = \Eds [ Y_{0,j} | \Yt_j = 0].
\end{equation}
\item $k_j$: By the definition of conditional expectation:
 \begin{align*}
    k_j = \Pds [Y_j = 0 |\Yt_j = 0 ] 
    = \frac{ \Pds [Y_j = 0 \cap \Yt_j = 0 ]}{\Pds [\Yt_j = 0] }.
\end{align*}
 The numerator is 
 \begin{align*}
 \Pds [Y_j = 0 \cap \Yt_j = 0 ] &= \Pds [Y_j = 0] \\ &= 1 - p_j,
\end{align*} 
where $p_j = \Pds [ Y_j \neq 0] = \Eds [M_{\Omega,jj}] $ is the probability that $j \in \Omega$. By the partition theorem, the denominator is 
 \begin{align*}
    \Pds [\Yt_j = 0] =& \Pds [\Yt_j = 0 | Y_j = 0]\Pds[Y_j = 0] \\
    &+ \Pds [\Yt_j = 0 | Y_j \neq 0]\Pds[Y_j \neq 0] \\
    =&  1 \cdot (1 - p_j) + (1 -  \pt_j) p_j \\
    =& 1 - \pt_j p_j,
\end{align*}
 where $\pt_j = \Pds [ \Yt \neq 0] = \Eds [M_{\Lambda, jj}] $. Therefore
 \begin{align}
     k_j = \Pds [Y_j = 0 |\Yt_j = 0 ] = \frac{1 - p_j}{1 - \pt_j p_j}. \label{eqn:k_def}
\end{align}
\end{itemize}
Substituting the above results into \eqref{eqn:Eyyt0} gives
\begin{align}
    \Eds[Y_j|\Yt_j = 0] = \Eds [Y_{0,j} | \Yt_j = 0] (1 - k_j), \label{eqn:case2end}
\end{align}
where $k_j$ is defined in \eqref{eqn:k_def}.

\textit{Combining Cases 1 and 2} ($\Eds[Y_j|\Yt_j]$): To find $\Eds[Y_j|\Yt_j]$, one must construct an expression that holds for both  \eqref{eqn:case1end} and \eqref{eqn:case2end}. Consider the following candidate:
\begin{align}
    \Eds[Y_j|\Yt_j]  = (1 - k_j) \Eds [Y_{0,j} | \Yt_j] + k_j \Yt_j. \label{eqn:Ey_yt}
\end{align}
This expression can be verified as consistent with \eqref{eqn:case1end} by setting $\Yt_j \neq 0$: 
\begin{align*}
	\Eds[Y_j|\Yt_j \neq 0] &= (1 - k_j) \Eds [Y_{0,j} | \Yt_j \neq 0]+ k_j \Yt_j \\
	&= (1 - k_j) \Yt_j + k_j \Yt_j \\
	&= \Yt_j,
\end{align*}
as required. Secondly, setting $\Yt_j = 0$ gives 
\begin{align*}
	\Eds[Y_j|\Yt_j = 0] &= (1 - k_j) \Eds [Y_{0,j} | \Yt_j = 0]  + k_j \cdot 0 \\
	&= (1 - k_j) \Eds [Y_{0,j} | \Yt_j = 0],
\end{align*}
as required by \eqref{eqn:case2end}. Therefore \eqref{eqn:Ey_yt} is consistent with both \eqref{eqn:case1end} and \eqref{eqn:case2end}, so is a correct expression for $\Eds[Y_j|\Yt_j]$.

When $1 - k_j \neq 0$ we can rearrange \eqref{eqn:Ey_yt} for $\Eds [Y_{0,j} | \Yt_j]$:
\begin{align}
	\mathds{E}[ Y_{0,j} | \Yt_j ] = (1 - k_j)^{-1} (\mathds{E}[ Y_j | \widetilde{Y}_j] - k_j \Yt_j). \label{eqn:exp_kj}
\end{align}
By the expression for $k_j$ given in \eqref{eqn:k_def},  $1 - k_j$ is
\begin{align*}
1 - k_j & = 1 - \frac{1 - p_j}{1 - \pt_j p_j}  = \frac{p_j(1 - \pt_j)}{1 - \pt_j p_j},
\end{align*}
so is non-zero when $p_j \neq 0$ and $\pt_j \neq 1$. Writing \eqref{eqn:exp_kj} in terms of vectors and matrices yields  \eqref{eqn:ey0}, as required.
\end{proof}
\subsection{Proof of SSDU \label{app:ssdu}}

This appendix proves that  a network trained with SSDU's loss weighting $(\eye - \Mt)\Mo$ satisfies
\begin{align}
(\eye - K)(\eye -\Mt \Mo) ( f_{\theta^*}(\Yt) - \Eds [Y_0 | \Yt] )  = 0.
\end{align}
\begin{proof}  By \eqref{eqn:whw_der}, the minimum of SSDU's loss function \eqref{eqn:LSSDU} gives a function that satisfies
\begin{align}
\Eds [(\eye - \Mt)\Mo(f_{\theta^*}(\Yt) -Y) | \Yt] = 0 \label{eqn:bob}
\end{align}
Similarly to Appendix A, the following derives expressions for $\Eds [(\eye - \Mt)\Mo(f_{\theta^*}(\Yt) -Y) | \Yt]$ under two conditions, $\Yt_j \neq 0$ and $\Yt_j = 0$, and subsequently find an expression that is true for both. In the following, $\mt_j$ and $m_j$ are the $j$th diagonals of $\Mt$ and $\Mo$ respectively. 

\textit{Case 1} ($\Eds [(1- \mt_j)m_j(f_{\theta^*}(\Yt)_j -Y_j) | \Yt_j \neq 0]$):
When $\Yt_j \neq 0$, the $j$th entry is not masked: $\mt_j = 1$. Therefore $(1 - \mt_j)m_j = 0$ and the expression is zero:
\begin{align}
\Eds [(1 - \mt_j)m_j(f_{\theta^*}(\Yt)_j -Y_j) | \Yt_j \neq 0] = 0. \label{eqn:app_B_cs_1}
\end{align}

\textit{Case 2} ($\Eds [(1- \mt_j)m_j(f_{\theta^*}(\Yt)_j -Y_j) | \Yt_j = 0]$): When $\Yt_j = 0$, $\mt_jm_j = 0$, so $(1 - \mt_j)m_j = m_j$. Therefore
\begin{multline}
	\Eds [(1 - \mt_j)m_j (f_{\theta^*}(\Yt)_j - Y_j)| \Yt_j = 0] \\ = \Eds [m_j (f_{\theta^*}(\Yt)_j - Y_j)| \Yt_j = 0]. \label{eqn:ll}
\end{multline}
As for Case 2 of Appendix A, we can use the partition theorem to express \eqref{eqn:ll} as a weighted sum:
\begin{align}
 \Eds [  & m_j   (  f_{\theta^*}(\Yt)_j  - Y_j)| \Yt_j = 0]  \nonumber \\ & =  \Eds [m_j (f_{\theta^*}(\Yt)_j - Y_j)| \Yt_j = 0 \cap Y_j = 0 ] \cdot k_j \nonumber \\ 
    & + \Eds [m_j (f_{\theta^*}(\Yt)_j - Y_j) | \Yt_j = 0 \cap Y_j \neq 0 ] \cdot (1 - k_j), \label{eqn:al}
\end{align}
where $k_j = \Pds [Y_j = 0 |\Yt_j = 0 ]$ as in Appendix A, given in \eqref{eqn:k_def}. Taking each term in turn:
\begin{itemize}
\item $\Eds [m_j (f_{\theta^*}(\Yt)_j - Y_j)| \Yt_j = 0 \cap Y_j = 0 ]$: Since $Y_j = 0$ when it is zeroed by the mask, $m_j = 0$. Therefore
\begin{equation*}
\Eds [m_j (f_{\theta^*}(\Yt)_j - Y_j)| \Yt_j = 0 \cap Y_j = 0 ] = 0.
\end{equation*}
\item $\Eds [m_j(f_{\theta^*}(\Yt)_j - Y_j)| \Yt_j = 0 \cap Y_j \neq 0 ]$: When $Y_j \neq 0$, it is not zeroed by the mask, so $m_j = 1$:
\begin{multline*}
\Eds [m_j(f_{\theta^*}(\Yt)_j - Y_j)| \Yt_j = 0 \cap Y_j \neq 0 ] \\ = \Eds [f_{\theta^*}(\Yt)_j - Y_j| \Yt_j = 0 \cap Y_j \neq 0 ].
\end{multline*}
Further, since $Y_j = Y_{0,j}$ when $Y_j \neq 0$ by the measurement model, 
\begin{multline*}
\Eds [f_{\theta^*}(\Yt)_j - Y_j| \Yt_j = 0 \cap Y_j \neq 0 ] \\ = \Eds [f_{\theta^*}(\Yt)_j - Y_{0,j}| \Yt_j = 0].
\end{multline*}
\end{itemize}    
Substituting the above results in to  \eqref{eqn:al} gives
    \begin{multline}
     \Eds [(1 - \mt_j)m_j (f_{\theta^*}(\Yt)_j - Y_j)| \Yt_j = 0] \\ 
    =  \Eds [f_{\theta^*}(\Yt)_j - Y_{0,j}| \Yt_j = 0 ] \cdot (1 - k_j). \label{eqn:app_B_cs_2}
\end{multline}

\textit{Combining Cases 1 and 2}: A correct expression for $\Eds [(1- \mt_j) m_j (f_{\theta^*}(\Yt)_j - Y_j)| \Yt_j = 0]$ must be true for both Case 1 and 2, so consistent with both \eqref{eqn:app_B_cs_1} and \eqref{eqn:app_B_cs_2}. Consider the candidate
\begin{multline}
\Eds [(1- \mt_j)m_j (f_{\theta^*}(\Yt)_j - Y_{0,j})| \Yt_j] \\ = (1 - k_j) (1 -\mt_j m_j)  \Eds [f_{\theta^*}(\Yt)_j - Y_{0,j}| \Yt_j]. \label{eqn:app_b_comb}
\end{multline}
Eqn. \eqref{eqn:app_b_comb} is consistent with \eqref{eqn:app_B_cs_1} because $(1 -\mt_j m_j) = 0$ when  $ \Yt_j \neq 0$, and consistent with \eqref{eqn:app_B_cs_2} because $(1 -\mt_j m_j) = 1$ when $\Yt_j = 0$. Using the vector form of \eqref{eqn:app_b_comb} and setting $\Eds [f_{\theta^*}(\Yt) | \Yt] = f_{\theta^*}(\Yt)$ gives 
\begin{multline}
	\Eds [(\eye - \Mt)\Mo(f_{\theta^*}(\Yt) -Y) | \Yt] \\ = (\eye - K)(\eye -\Mt \Mo) ( f_{\theta^*}(\Yt) - \Eds [Y_0 | \Yt] )  = 0, \label{eqn:1mk}
\end{multline} 
as required.
\end{proof}

\bibliographystyle{ieeetr}
\bibliography{library_manu}

\end{document}


\maketitle

\section*{Results on 1D sampled brain data at $R = 4$}

Figures S\ref{fig:bar} and S\ref{fig:box_plots} show the methods' performance on the test set when $\Rt$ is tuned. The medians are shown in Table S1. Like figures 3 and 4 of the main paper, the distribution closest to fully supervised is K-weighted 1D partitioned SSDU with $y_s$ input. 

Fig. S\ref{fig:4x} shows a reconstruction example from the test set. Although 1D partitioned SSDU outperforms 2D partitioned SSDU quantitatively, both offer high quality reconstructions with minimal reconstruction artifacts. This contrasts with Fig. 4 of the main paper, where there is a clear qualitative difference between the reconstruction quality. This indicates that the benefit of using the same type  of distribution of $\Mt$ as $\Mo$ is less substantial for less ambitious acceleration factors. Further, K-weighting SSDU only marginally improves the method's performance qualitatively. This is due to $(\eye - K)^{-1}$ have smaller entries when $R_\Omega = 4$, so is a less crucial sampling compensation.

Fig. S\ref{fig:line_4x} shows the dependence of the self-supervised methods on $\Rt$ at $R_\Omega=4$. As in the $R_\Omega=8$ plot in Fig. 6 of the main paper, K-weighted 1D partitioned SSDU is the most robust, and 2D partitioned SSDU is very sensitive to the tuning of $\Rt$, especially for the estimate with $y_s$ input.

\begin{figure}
    \includegraphics[width=0.5\textwidth]{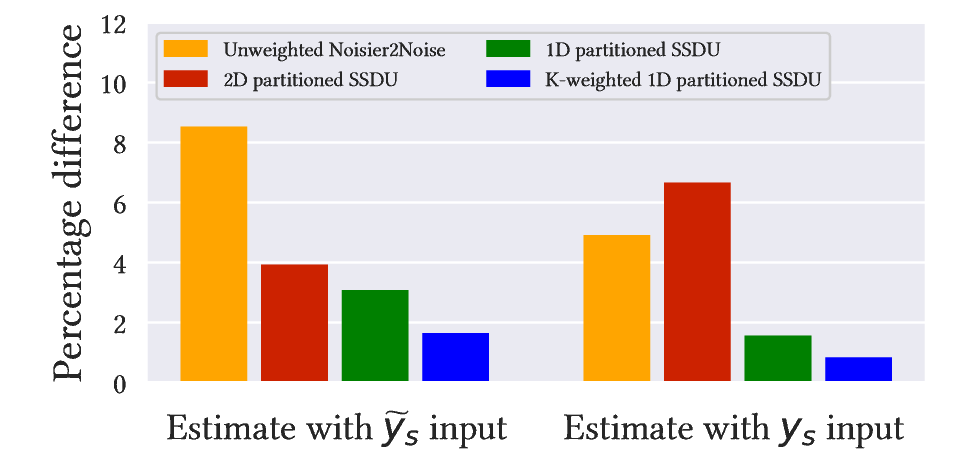}
    \caption{The percentage difference between fully supervised method and the mean NMSE for all methods at $R_\Omega=4$. Here, $\Mo$ is 1D distributed and $\Rt$ has been tuned to minimize the test set NMSE. }
    \label{fig:bar}

	\centering
    \includegraphics[width=0.5\textwidth]{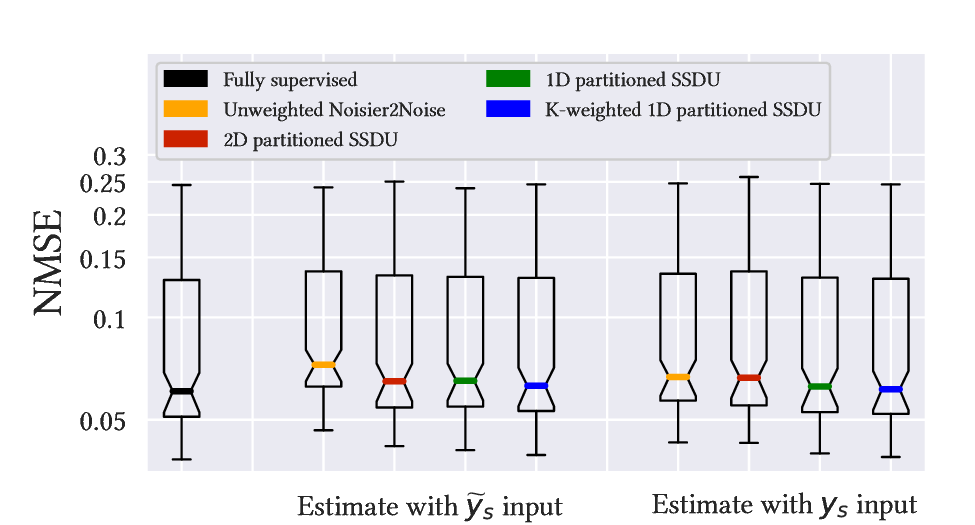}
    \caption{The NMSE and SSIM for all methods at $R_\Omega=4$, where $\Rt$ has been tuned to minimize the test set NMSE. 1D partitioned SSDU was found to perform the most similarly to fully supervised training. The exact numerical values are in Table S1.}
    \label{fig:box_plots}
\end{figure} 

 \begin{figure}[]
    \centering
    \includegraphics[width=0.5\textwidth]{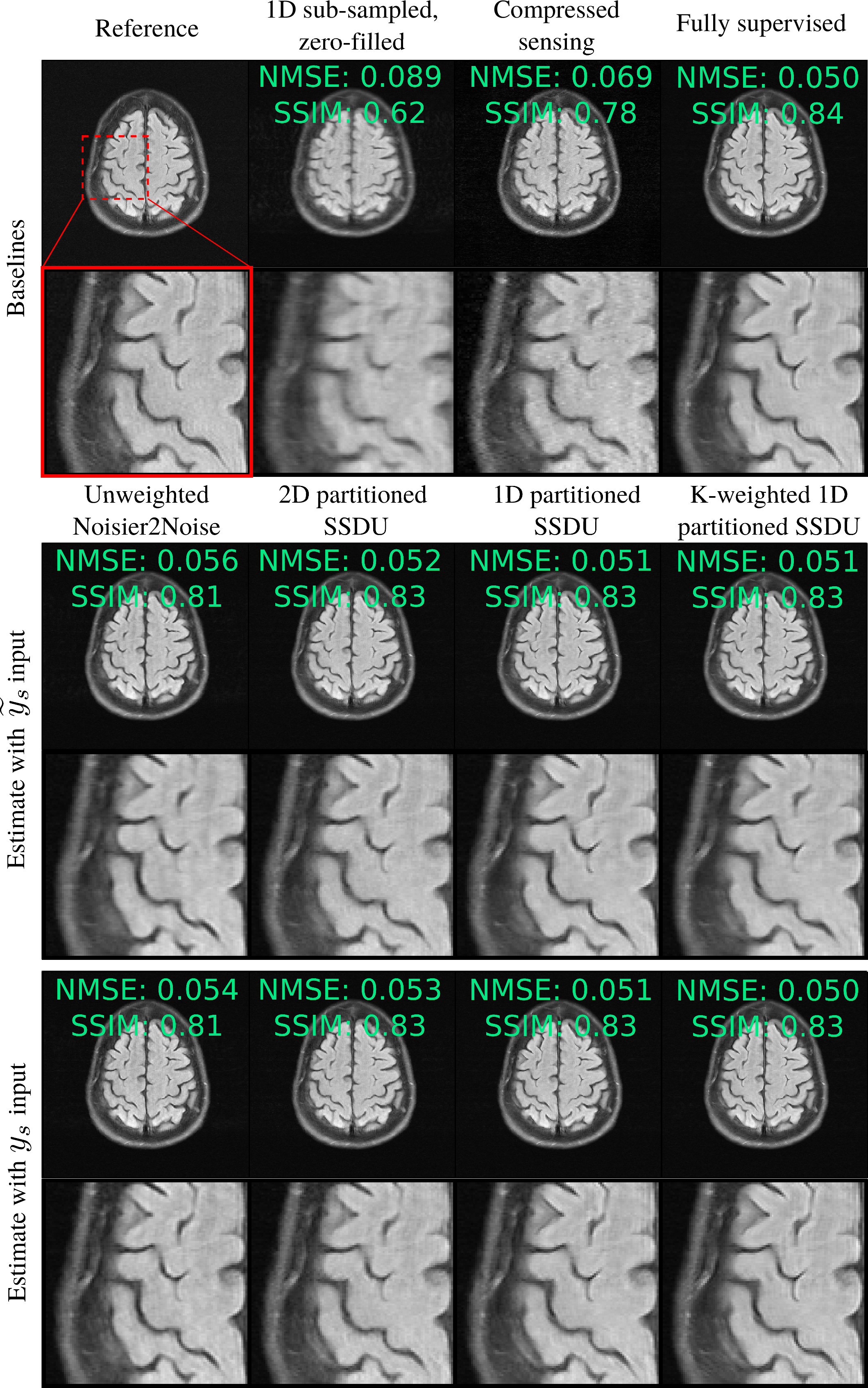}
    \caption{A reconstruction example at $R_\Omega=4$, with a tuned $\Rt$. As in Fig. 5 of the main paper, K-weighted 1D partitioned SSDU has the best score. However, there is a less substantial difference between the estimates. }
    \label{fig:4x}

    \centering
    \includegraphics[width=0.48\textwidth]{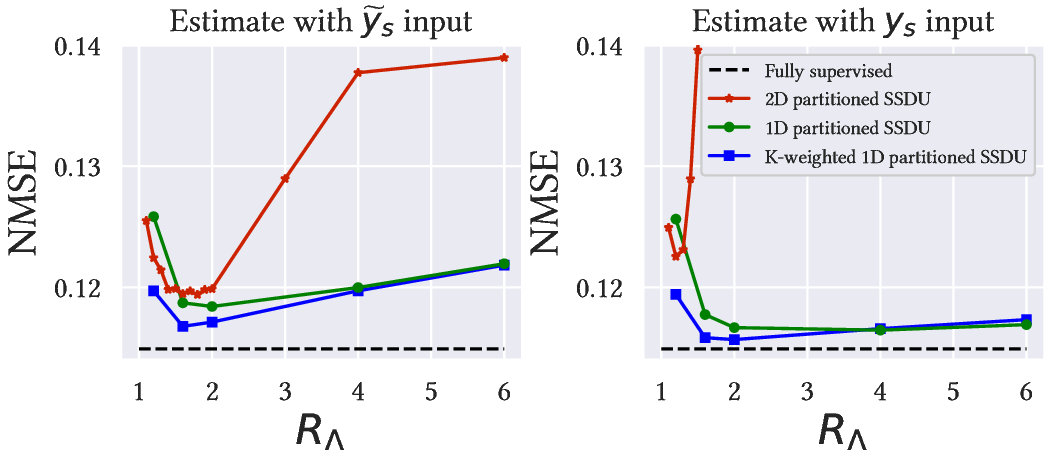}
    \caption{The dependence of the test set NMSE on $\Rt$ at $R_\Omega=4$ for both outputs. }
    \label{fig:line_4x}
\end{figure}

\newpage  
\section*{Table of median NMSE results on the test set}

For reference, Table S1 shows the median NMSE scores on the test set for tuned $\Rt$. In other words, it shows the numerical values of the colored horizontal lines in Fig. 4 of the main manuscript and Fig. S2 for $R_\Omega = 8,4$ respectively.


%
%

\begin{table}[]
\centering
\scriptsize
\hspace{0.8cm}
{ 
\begin{tabular}{c||c|c||c|c}
\textsc{Sub-sampling factor} $R_\Omega$        & \multicolumn{2}{c||}{4} &  \multicolumn{2}{c}{8}  \\ \hline
\textsc{Network input at inference}     & \multicolumn{1}{c|}{$\yt_s$} & \multicolumn{1}{c||}{$y_s$} & \multicolumn{1}{c|}{$\yt_s$} & \multicolumn{1}{c}{$y_s$} \\ \hhline{=#=|=#=|=}
Fully supervised          & -          & 0.061          & -          & 0.077       \\ \hline
Unweighted Noisier2Noise        & 0.073                   & 0.067 & 0.096                   & 0.097                         \\ \hline
2D partitioned SSDU & 0.065                   & 0.066                  & 0.083                & 0.089                \\ \hline
1D partitioned SSDU       & 0.067               & 0.063 & 0.084     & \textbf{0.078} \\ \hline
K-weighted 1D partitioned SSDU       & 0.063                   & \textbf{0.061} & 0.082 & \textbf{0.078}                        \\
\end{tabular}
}
\caption*{Table S1: The median NMSE on the test set for all methods, where $\Rt$ is tuned. The best result from the self-supervised methods is highlighted in bold.}
\end{table}